\documentclass[11pt, times]{article}
\usepackage{times}
\usepackage[T1]{fontenc}
\usepackage[latin1]{inputenc}
\usepackage{epic}
\usepackage{eepic}
\usepackage{alltt}
\usepackage{epsfig,psfig}
\usepackage{graphicx}
\usepackage{amssymb,amsmath,amsthm,algorithm}
\input{psfig.sty}

\ifx\pdftexversion\undefined
  \usepackage[colorlinks,linkcolor=black,filecolor=black,
  citecolor=black,urlcolor=black,pdfstartview=FitH]{hyperref}
\else
  \usepackage[colorlinks,linkcolor=blue,filecolor=blue,
  citecolor=blue,urlcolor=blue,pdfstartview=FitH]{hyperref}
\fi

\newcommand{\tb}{\makebox[0.5cm]{}}
\newcommand{\due}{\makebox[1.2cm]{}}
\newcommand{\tre}{\makebox[1.9cm]{}}

\newcommand{\dd}[1]{{\sffamily\bfseries\Large \mbox{***DD:\ } #1 \mbox{****
}}}
\newcommand{\ad}[1]{{\sffamily\bfseries\Large \mbox{***AD:\ } #1 \mbox{****
}}}

\newcommand{\ra}{\rangle}
\newcommand{\la}{\langle}
\newcommand\pulse{\mbox{\slshape pulse}}
\newcommand\cycle{\mbox{\slshape cycle}}
\newcommand\Cycle{\mbox{\slshape Cycle}}

\newcommand\MessagePool{\mbox{\slshape Message\_Pool}}
\newcommand\cyclemin{\mbox{$cycle_{\mbox{\scriptsize\slshape min}}$}}
\newcommand\cyclemax{\mbox{$cycle_{\mbox{\scriptsize\slshape max}}$}}

\newcommand\initiator{\mbox{\sc Initiator-Accept\ }}
\newcommand\ginit{\mbox{{\sc I}-accept}}

\def\squarebox#1{\hbox to #1{\hfill\vbox to #1{\vfill}}}

\newcommand{\ignore}[1]{}
\newsavebox{\theorembox}
\newsavebox{\lemmabox}
\newsavebox{\conjecturebox}
\newsavebox{\claimbox}
\newsavebox{\factbox}
\newsavebox{\corollarybox}
\newsavebox{\propositionbox}
\newsavebox{\examplebox}
\savebox{\theorembox}{\bf Theorem} \savebox{\lemmabox}{\bf Lemma}
\savebox{\conjecturebox}{\bf Conjecture} \savebox{\claimbox}{\bf
Claim} \savebox{\factbox}{\bf Fact} \savebox{\corollarybox}{\bf
Corollary} \savebox{\propositionbox}{\bf Proposition}
\savebox{\examplebox}{\bf Example}
\newtheorem{theorem}{\usebox{\theorembox}}
\newtheorem{lemma}{\usebox{\lemmabox}}[section]

\newtheorem{corollary}[lemma]{\usebox{\corollarybox}}

\newtheorem{definition}{{\sc Definition}\rm }[section]

\newtheorem{observation}{{\sc Observation}\rm }[section]

\begin{document}

%\mainmatter
\title{Self-stabilizing Pulse Synchronization Inspired by Biological Pacemaker Networks}

%\titlerunning{Self-stabilizing Pulse Synchronization}

\ignore{
\author{ Ariel Daliot\inst{1} \and Danny
Dolev\inst{1} \and Hanna Parnas\inst{2}}
\authorrunning{Daliot, Dolev and Parnas}

\institute{School of Engineering and Computer Science, The Hebrew
University of Jerusalem, Israel.
\email{\{adaliot,dolev\}@cs.huji.ac.il} \and Department of
Neurobiology and the Otto Loewi Minerva Center for Cellular and
Molecular Neurobiology, Institute of Life Science, The Hebrew
University of Jerusalem, Israel.  \email{hanna@vms.huji.ac.il} }
}

\author{ Ariel Daliot\inst{1} \and Danny
Dolev\inst{1} \and  Hanna Parnas\inst{2}}

\author{
Ariel Daliot\thanks{School of Engineering and Computer Science. The
Hebrew University of Jerusalem, Israel. Email:
adaliot@cs.huji.ac.il}, Danny Dolev\thanks{School of Engineering and
Computer Science. The Hebrew University of Jerusalem, Israel. Email:
dolev@cs.huji.ac.il} and Hanna Parnas\thanks{Department of
Neurobiology and the Otto Loewi Minerva Center for Cellular and
Molecular Neurobiology, Institute of Life Science, The Hebrew
University of Jerusalem, Israel. Email: hanna.parnas@huji.ac.il }}

\date{}

%\authorrunning{Daliot, Dolev and Parnas}

%\institute{School of Engineering and Computer Science, The Hebrew
%University of Jerusalem, Israel
%\email{\{adaliot,dolev\}@cs.huji.ac.il} \and Department of
%Neurobiology and the Otto Loewi Minerva Center for Cellular and
%Molecular Neurobiology, Institute of Life Science, The Hebrew
%University of Jerusalem, Israel  \email{hanna.parnas@huji.ac.il} }

%\date{Feb. 28, 2003 (revised)}
\maketitle

%\pagestyle{empty} \pagenumbering{roman}
%\thepage{ }
%\pagestyle{headings}

\begin{abstract}
We define the ``Pulse Synchronization'' problem that requires nodes
to achieve tight synchronization of regular pulse events, in the
settings of distributed computing systems. Pulse-coupled
synchronization is a phenomenon displayed by a large variety of
biological systems, typically overcoming a high level of noise.
Inspired by such biological models, a robust and self-stabilizing
Byzantine pulse synchronization algorithm for distributed computer systems is
presented. The algorithm attains near optimal synchronization
tightness while tolerating up to a third of the nodes exhibiting
Byzantine behavior concurrently. Pulse synchronization has been
previously shown to be a powerful building block for designing
algorithms in this severe fault model. We have previously shown how
to stabilize general Byzantine algorithms, using pulse
synchronization. To the best of our knowledge there is no other
scheme to do this without the use of synchronized pulses.
\end{abstract}


%\newpage
\pagestyle{plain} \pagenumbering{arabic}

\section{Introduction}
\label{sec:s1} \ignore{-2mm} The phenomenon of synchronization is
displayed by many biological systems~\cite{R4}. It presumably plays
an important role in these systems. For example, the heart of the
lobster is regularly activated by the synchronized firing of four
interneurons in the cardiac pacemaker network~\cite{R16,R17}. It was
concluded that the organism cannot survive if all four interneurons
fire out of synchrony for prolonged times~\cite{R2}. This system
inspired the present work. Other examples of biological
synchronization include the {\em malaccae} fireflies in Southeast
Asia where thousands of male fireflies congregate in mangrove trees,
flashing in synchrony~\cite{R13}; oscillations of the neurons in the
circadian pacemaker, determining the day-night rhythm; crickets that
chirp in unison~\cite{R14}; coordinated mass spawning in corals and
even audience clapping together after a ``good''
performance~\cite{R15}. Synchronization in these systems is
typically attained despite the inherent variations among the
participating elements, or the presence of noise from external
sources or from participating elements. A generic mathematical model
for synchronous firing of biological oscillators based on a model of
the human cardiac pacemaker is given in~\cite{R1}. This model does
not account for noise or for the inherent differences among
biological elements.

In computer science, synchronization is both a goal by itself and a
building block for algorithms that solve other problems. In the
``Clock Synchronization'' problem, it is required of computers to
have their clocks set as close as possible to each other as well as
to keep a notion of real-time
(\cite{DHSS95,Lamport:1985:SCP,Liskov:1991}).

In general, it is desired for algorithms to guarantee correct
behavior of the system in face of faults or failing elements,
without strong assumptions on the initial state of the system. It
has been suggested in~\cite{R2} that similar fault considerations
may have been involved in the evolution of distributed biological
systems. In the example of the cardiac pacemaker network of the
lobster, it was concluded that at least four neurons are needed in
order to overcome the presence of one faulty neuron, though
supposedly one neuron suffices to activate the heart. The cardiac
pacemaker network must be able to adjust the pace of the
synchronized firing according to the required heartbeat, up to a
certain bound, without losing the synchrony (e.g. while escaping a
predator a higher heartbeat is required -- though not too high). Due
to the vitality of this network, it is presumably optimized for
fault tolerance, self-stabilization, tight synchronization and for
fast re-synchronization.

The apparent resemblance of the synchronization and fault tolerance
requirements of biological networks and distributed computer
networks makes it very appealing to infer from models of biological
systems onto the design of distributed algorithms in computer
science. Especially when assuming that distributed biological
networks have evolved over time to particularly tolerate inherent
heterogeneity of the cells, noise and cell death. In the current
paper, we show that in spite of obvious differences, a biological
fault tolerant synchronization model (\cite{R2}) can inspire a novel
solution to an apparently similar problem in computer science.

We propose a relaxed version of the Clock Synchronization problem,
which we call ``Pulse Synchronization'', in which all the elements
are required to invoke some regular pulse (or perform a ``task'') in
tight synchrony, but allows to deviate from exact regularity. Though
nodes need to invoke the pulses synchronously, there is a limit on
how frequently it is allowed to be invoked (similar to the linear
envelope clock synchronization limitation). The ``Pulse
Synchronization'' problem resembles physical/biological
pulse-coupled synchronization models~\cite{R1}, though in a computer
system setting an algorithm needs to be supplied for the nodes to
reach the synchronization requirement. To the best of our knowledge
this problem has not been formally defined in the settings of
distributed computer systems.

We present a novel algorithm in the settings of self-stabilizing
distributed algorithms, instructing the nodes how and when to invoke
a pulse in order to meet the synchronization requirements of ``Pulse
Synchronization''. The core elements of the algorithm are analogous
to the neurobiological principles of {\em endogenous}\/ (self
generated) {\em periodic spiking, summation} and {\em time dependent
refractoriness}. The basic algorithm is quite simple: every node
invokes a pulse regularly and sends a message upon invoking it ({\em
endogenous periodic spiking}). The node sums messages received in
some ``window of time'' ({\em summation}) and compares this to the
continuously decreasing time dependent firing threshold for invoking
the pulse ({\em time dependent refractory function}). The node fires
when the counter of the summed messages crosses the current
threshold level, and then resets its cycle. For in-depth
explanations of these neurobiological terms see~\cite{BioPhysKoch}.

The algorithm performs correctly as long as less than a third of the
nodes behave in a completely arbitrary (``Byzantine'') manner
concurrently. It ensures a tight synchronization of the pulses of
all correct nodes, while not using any central clock or global
pulse. We assume the communication network allows for a broadcast
environment and has a bounded delay on message transmission. The
algorithm may not reach its goal as long as these limitations are
violated or the network graph is disconnected. The algorithm is
self-stabilizing Byzantine and thus copes with a more severe fault
model than the traditional Byzantine fault model. Classic Byzantine
algorithms, which are not designed with self-stabilization in mind,
typically make use of assumptions on the initial state of the system
such as assuming all clocks are initially synchronized,
(c.f.~\cite{DHSS95}). Observe that the system might temporarily be
thrown out of the assumption boundaries, e.g. when more than one
third of the nodes are Byzantine or messages of correct nodes get
lost. When the system eventually returns to behave according to
these presumed assumptions it may be in an arbitrary state. A
classic Byzantine algorithm, being non-stabilizing, might not
recover from this state. On the other hand, a self-stabilizing
protocol converges to its goal from any state once the system
behaves well again, but is typically not resilient to permanent
faults. For our protocol, once the system complies with the
theoretically required bound of $f<3n$ permanent Byzantine faulty
nodes in a network of $n$ nodes then, regardless of the state of the
system, tight pulse synchronization is achieved within finite time.
It overcomes transient failures and permanent Byzantine faults and
makes no assumptions on any initial synchronized activity among the
nodes (such as having a common reference to time or a common event
for triggering initialization).

Our algorithm is uniform, all nodes execute an identical algorithm.
It does not suffer from communication deadlock, as can happen in
message-driven algorithms (\cite{R27}), since the nodes have a
time-dependent state change, at the end of which they fire
endogenously. The faulty nodes cannot ruin an already attained
synchronization; in the worst case, they can slow down the
convergence towards synchronization and speed up the synchronized
firing frequency up to a certain bound. The convergence time is
$O(f)$ cycles with a near optimal synchronization of the pulses to
within  $d$ real-time (the bound on the end to end network and
processing delay). We show in Subsection~\ref{sec:IAGree} how the
algorithm can be executed in a non-broadcast network to achieve
synchronization of the pulses to within $3d$ real-time.

\ignore{Comparatively, the convergence time of the digital clock
synchronization problem in a similar model presented
in~\cite{DolWelSSBYZCS04} is exponential in the network size. The
synchronization tightness achieved there depends on the network
size. To the best of our knowledge this is the only paper describing
an algorithm for Byzantine self-stabilizing clock synchronization
algorithm. Note also that there are many papers that deal with
self-stabilizing clock synchronization
(see~\cite{H00b,DW97b,ADG91}), though not facing Byzantine faults.
Nonetheless, the convergence time in these papers is not linear.\\}

{\bf Applications and contribution of this paper:} We have shown in
\cite{ByzStabilizer} how to stabilize general Byzantine algorithms
using synchronized pulses. In \cite{token-tr} we have presented a
very efficient, besides being the first, self-stabilizing Byzantine
token passing algorithm. The efficient self-stabilizing Byzantine
clock synchronization algorithm in \cite{DDPBYZ-CS03} is also the
first such algorithm for clock synchronization. All these algorithms
assume a background self-stabilizing Byzantine pulse synchronization
module though the particular pulse synchronization procedure
presented in \cite{DDPBYZ-CS03} suffers from a flaw\footnote{The
flaw was pointed out by Mahyar Malekpour from NASA LaRC and Radu
Siminiceanu from NIA, see \cite{NASA-BYZSS} .}. The only other
self-stabilizing Byzantine pulse synchronization algorithm (besides
the current work), is to the best of our knowledge, the one in
\cite{new-pulse-tr}, which is a correction to the one in
\cite{DDPBYZ-CS03}. In comparison to the current paper, the pulse
synchronization algorithm in \cite{new-pulse-tr} has a much higher
message complexity and worse tightness, is more complicated but it
converges in $O(1),$ does not assume broadcast and scales better.
The current paper is simpler, uses much shorter messages; it has a
smaller message complexity and introduces novel and interesting
elements to distributed computing.

In the Discussion, in Section~\ref{sec:discussion}, we postulate
that our result elucidates the feasibility and adds a solid brick to
the motivation to search for and to understand biological mechanisms
for robustness that can be carried over to computer systems.

\section{Model and Problem Definition}
\label{sec:model}\ignore{-2mm}

The environment is a network of $n$ nodes, out of which $f$ are
faulty nodes, that communicate by exchanging messages. The nodes
regularly invoke ``pulses'', ideally every $\Cycle$ real-time units.
The invocation of the pulse is expressed by sending a message to all
the nodes; this is also referred to as {\bf firing}. We assume that
the message passing allows for an authenticated identity of the
senders. The communication network does not guarantee any order on
messages among different nodes. Individual nodes have no access to a
central clock and there is no external pulse system. The hardware
clock rate (referred to as the {\em physical timers}) of correct
nodes has a bounded drift, $\rho,$ from real-time rate.  When the
system is not coherent then there can be an unbounded number of
concurrent Byzantine faulty nodes, the turnover rate between faulty
and non-faulty nodes can be arbitrarily large and the communication
network may behave arbitrarily. Eventually the system settles down
in a coherent state in which there at most $f<3n$ permanent
Byzantine faulty nodes and the communication network delivers
messages within bounded time.

\begin{definition} A node is {\bf non-faulty} at times that it complies with the
following:
\begin{enumerate}
\vspace{-0.5em} \item \emph{(Bounded Drift)} Obeys a global constant
$0<\rho<<1$ (typically $\rho \approx 10^{-6}$), such that for every
real-time interval $[u,v]:$\vspace{-2mm}
$$(1-\rho)(v-u)  \le \mbox{ `physical timer'}(v) -
\mbox{ `physical timer'}(u) \le (1+\rho)(v-u).$$\vspace{-6mm}

\item \emph{(Obedience)} Operates according to the correct protocol.

\item \emph{(Bounded Processing Time)} Processes any message of the correct protocol within $\pi$ real-time units of arrival time.

\end{enumerate}
\end{definition}

\ignore{

\ignore{Table~\ref{Table2}} Table 1 displays the actual time
differences counted by clocks exhibiting the allowed extremes of the
clock skew.

\begin{table}
\label{Table2}
 \caption{\footnotesize Time count conversion
table between nodes displaying the upper and lower bounds of the
allowed clock skews (fastest and slowest nodes respectively). For
example, suppose the first row corresponds to a time frame of $1000$
real-time units. A hardware clock whose skew is at the lower end
(slowest) of the allowed range will count only $1000000(1-10^{-6}) =
999999$ time units, while a hardware clock with the highest
acceptable skew (fastest) will count $1000000(1+10^{-6}) = 1000001$
time units accordingly.}\begin{minipage}{122mm} \center{
\begin{tabular}{|c|c|c|} \hline
 Slowest node  &  Real time  & Fastest node \\ \hline
$x(1-\rho)$ & $x$ &$x(1+\rho)$ \\ \hline $x$ & $\frac{x}{1-\rho}$ &
$\frac{x(1+\rho)}{1-\rho}$
\\ \hline
$\frac{x(1-\rho)}{1+\rho}$ & $\frac{x}{1+\rho}$& $x$ \\ \hline
\end{tabular}
}

\end{minipage}
\end{table}
}

A node is considered {\bf faulty} if it violates any of the above
conditions. The faulty nodes can be Byzantine. A faulty node may
recover from its faulty behavior once it resumes obeying the
conditions of a non-faulty node. In order to keep the definitions
consistent the ``correction'' is not immediate but rather takes a
certain amount of time during which the non-faulty node is still not
counted as a correct node, although it supposedly behaves
``correctly''\footnote{For example, a node may recover with
arbitrary variables, which may violate the validity condition if
considered correct immediately.}. We later specify the time-length
of continuous non-faulty behavior required of a recovering node to
be considered \textbf{correct}.

\begin{definition} The communication network is {\bf non-faulty} at periods that it complies with the
following: \label{def:net_nonfaulty}
\begin{itemize}

\item \emph{(Bounded Transmission Delay)} Any message sent or received by a non-faulty node will
arrive at every non-faulty node within $\delta$ real-time units.
\end{itemize}
\end{definition}

Thus, our communication network model is an ``eventual bounded-delay''
communication network.

\vspace{+2mm} \noindent{\bf Basic definitions and notations: }\\

We use the following notations though nodes do not need to maintain
all of them as variables.
\begin{itemize}
\vspace{-0.5em}

\item $d\equiv \delta + \pi.$ Thus, when the communication network is non-faulty,
$d$ is the upper bound on the elapsed real-time from the sending of
a message by a non-faulty node until it is received and processed by
every correct node.

\item A \textbf{$\pulse$} is an internal event targeted to happen in
``tight''\footnote{We consider $c\cdot d, $ for some small constant
$c,$ as tight.} synchrony at all correct nodes. A \textbf{$\Cycle$}
is the ``ideal'' time interval length between two successive pulses
that a node invokes, as given by the user. The actual \cycle\
length, denoted in regular caption, has upper and lower bounds as a
result of faulty nodes and the physical clock skew.

\item $\sigma$ represents the upper bound on the real-time window within which
all correct nodes invoke a pulse ({\em tightness of pulse
synchronization}). Our solution achieves $\sigma=d.$ We assume that
$\Cycle \gg\sigma.$

\ignore{-3mm} \item $\phi_i(t) \in \mathbb{R}^+\cup \{\infty\},$
$0\le i\le n,$ denotes, at real-time $t,$ the elapsed real-time
since the last pulse invocation of $p_i.$ It is also denoted as the
``$\phi$ of node $p_i$''. We occasionally omit the reference to the
time in case it is clear out of the context. For a node, $p_j,$ that
has not fired since initialization of the system, $\phi_j \equiv
\infty.$

\item $\cyclemin$ and
$\cyclemax$ are values that define the bounds on the actual \cycle\
length during correct behavior. We achieve
$$\cycle_{\mbox{\scriptsize\slshape
min}}=\frac{n-2f}{n-f}\cdot\Cycle\cdot(1-\rho) \le \cycle \le
\Cycle\cdot(1+\rho)=\cyclemax\;.$$

\ignore{-3mm} \item $\mbox{\slshape message\_decay}\/$ represents
the maximal real-time a non-faulty node will keep a message or a
reference to it, before deleting it\footnote{The exact elapsed time
until deleting a messages is specified in the {\sc prune} procedure
in Fig.~\ref{alg:prune}.}.

\end{itemize}

In accordance with Definition~\ref{def:net_nonfaulty}, the network
model in this paper is such that every message sent or received by a non-faulty node arrives
within bounded time, $\delta,$ at all
non-faulty nodes. The algorithm and its respective proofs are specified in a stronger network
model in which every message received by a non-faulty node arrives
within $\delta$ time at all
non-faulty nodes. The subtle difference in the latter definition equals the assumption that every message received by a non-faulty node, even a message from a Byzantine node, will eventually reach every non-faulty node. This weakens the possibility for two-faced behavior by Byzantine nodes. The algorithm is able to utilize this fact so that if executed in such a network environment, then it can
attain a very tight, near optimal, pulse synchronization of $d$ real-time units. We show in Subsection~\ref{sec:IAGree} how to execute in the background a self-stabilizing Byzantine
reliable-broadcast-like primitive, which executes in the network model of Definition~\ref{def:net_nonfaulty}. This primitive effectively relays every message received by a non-faulty node so that the latter network model is satisfied. In such a case the algorithm can be executed in the network model of Definition~\ref{def:net_nonfaulty} and achieves synchronization of the pulses to within $3d$ real-time.

\ignore{The protocol below requires that $\cyclemin>(3n+1)\cdot
d\;+\;3d.$}

Note that the protocol parameters $n,$ $f$ and $\Cycle$ (as well as
the system characteristics $d$ and $\rho$) are fixed constants and
thus considered part of the incorruptible correct code\footnote{A
system cannot self-stabilize if the entire code space can be
perturbed, see \cite{CodeStabilizationSSS06}.}. Thus we assume that
non-faulty nodes do not hold arbitrary values of these constants.
\ignore{It is required that $\Cycle$ is chosen s.t. $\cyclemin$ is
large enough to allow our protocol to terminate in between pulses.}

A recovering node should be considered correct only once it has been
continuously non-faulty for enough time to enable it to have decayed
old messages and to have exchanged information with the other nodes
through at least a cycle.

\begin{definition} A node is {\bf correct} following $\cyclemax
+\sigma+ \mbox{\slshape message\_decay}$ real-time of continuous
non-faulty behavior.
\end{definition}

\begin{definition} The communication network is {\bf correct} following $\cyclemax
+\sigma+ \mbox{\slshape message\_decay}$ real-time of continuous
non-faulty behavior.
\end{definition}

\begin{definition}\label{def:system-coherence} \emph{(System Coherence)} The system is said to
be {\bf coherent} at times that it complies with the following:

\begin{enumerate}

\item \emph{(Quorum)} There are at least $n-f$ correct nodes, where $f$ is the
upper bound on the number of potentially non-correct nodes, at
steady state.

\item \emph{(Network Correctness)} The communication network is correct.
\end{enumerate}
\end{definition}

Hence, if the system is not coherent then there can be an unbounded
number of concurrent faulty nodes; the turnover rate between the
faulty and non-faulty nodes can be arbitrarily large and the
communication network may deliver messages with unbounded delays, if
at all. The system is considered coherent, once the communication
network and a sufficient fraction of the nodes have been non-faulty
for a sufficiently long time period for the pre-conditions for
convergence of the protocol to hold. The assumption in this paper,
as underlies any other self-stabilizing algorithm, is that the
system eventually becomes coherent.

All the lemmata, theorems, corollaries and definitions hold as long
as the system is coherent.\\

% begin ignore
\ignore{

\begin{itemize}
\item The {\bf pulse\_state} of the system at real-time $t$ is
given by: $$pulse\_state(t) \equiv (\phi_0(t), \ldots,
\phi_{n-1}(t))\;.$$

\item Let $G$ be the set of all possible pulse\_states of a system.

\item A set of nodes, $N,$ are called {\bf synchronized} at
real-time $t$ if $$\forall\, p_i, p_j \in N,
 |\phi_i(t) - \phi_j(t)| \le \sigma.$$

\item $s\in G$ is a {\bf synchronized\_pulse\_state}
\emph{of a set of nodes $N$}, at real-time $t$ if the set of correct
nodes in  $N$ are synchronized at some real-time $t_{syn} \in [t,
t+\sigma].$

\item $s\in G$ is a {\bf synchronized\_pulse\_state}
\emph{of the system} at real-time $t$ if the set of correct nodes
are synchronized at some real-time $t_{syn} \in [t, t+\sigma].$
\end{itemize}

The definition of a synchronized\_pulse\_state masks the recurring
brief time period in which a correct node in a synchronized set has
just fired while other are about to fire. In this short period the
nodes violate the definition of a synchronized set of nodes,
although the system may be in a legal state.

Note that if the set $N$ encompasses all the correct nodes then a
synchronized\_pulse\_state of the set of nodes $N,$ implies a
synchronized\_pulse\_state of the system. Observe that the
definition of a synchronized set of nodes is not transitive, thus
the definition of a synchronized\_pulse\_state of a set of nodes is
not transitive either.

} % end ignore

We now seek to give an accurate and formal definition of the notion
of pulse synchronization. We start by defining a subset of the
system states, which we call \emph{pulse\_states}, that are
determined only by the elapsed real-time since each individual node
invoked a pulse (the $\phi$'s). We then identify a subset of the
pulse\_states in which some set of correct nodes have ''tight`` or
''close`` $\phi$'s. We refer to such a set as a \emph{synchronized}
set of nodes. To complete the definition of synchrony there is a
need to address the recurring brief time period in which a correct
node in a synchronized set of nodes has just fired while others are
about to fire. This is addressed by adding to the definition nodes
whose $\phi$'s are almost a \Cycle\ apart.

If all correct nodes in the system comprise a synchronized set of
nodes then we say that the pulse\_state is a
\emph{synchronized\_pulse\_states of the system}. The objective of
the algorithm is hence to reach a synchronized\_pulse\_state of the
system and to stay in such a state. The methodology to prove that
our algorithm does exactly this will be to show firstly that a
synchronized set of correct nodes stay synchronized. Secondly, we
show that such synchronized sets of correct nodes incessantly join
together to form bigger synchronized sets of nodes. This goes on
until a synchronized set that encompasses all correct nodes in the
system is formed.

\begin{itemize}
\item The {\bf pulse\_state} of the system at real-time $t$ is
given by: $$pulse\_state(t) \equiv (\phi_0(t), \ldots,
\phi_{n-1}(t))\;.$$

\item Let $G$ be the set of all possible pulse\_states of a system.

\item A set of nodes, $S,$ is called {\bf
synchronized} at real-time $t$ if\\
\vspace{+1mm}$\forall p_i, p_j \in S,$  $\phi_i(t),\phi_j(t) \le
\cyclemax,$ and  one of the following is true:
\begin{enumerate}
\item

 $ |\phi_i(t) - \phi_j(t)| \le \sigma, \mbox{\ \ \ or }$

\item

$ \cyclemin - \sigma \le |\phi_i(t) - \phi_j(t)| \le \cyclemax$ and
$ |\phi_i(t-\sigma) - \phi_j(t-\sigma)| \le \sigma.$
\end{enumerate}

\item $s\in G$ is a {\bf synchronized\_pulse\_state}
\emph{of the system} at real-time $t$ if the set of correct nodes is
synchronized at real-time  $t.$
\end{itemize}

%\newpage
\begin{definition}\label{def:pulse-synch}{\bf The Self-Stabilizing Pulse Synchronization
Problem}\vspace{+2mm}\\

\vspace{-1mm}

\noindent{\bf Convergence:} Starting from an arbitrary system state,
 the system reaches a synchronized\_pulse\_state after a finite
time.\\

\noindent{\bf Closure:} If $s$ is a synchronized\_pulse\_state of
the system at real-time $t_0$ then $\forall\,$ real-time $t,t\ge
t_0,$
\begin{enumerate}
\vspace{-1mm} \item pulse\_state(t) is a synchronized\_pulse\_state,

\ignore{\item $(t-t_0\le \cycle_{\mbox{\scriptsize\slshape  min}}
\Rightarrow \psi_i(t,t_0)\le 1) \bigwedge (t-t_0\ge
\cycle_{\mbox{\scriptsize\slshape  max}} \Rightarrow
\psi_i(t,t_0)\ge 1),$\\ for every correct node $i.$}

\item In the real-time interval [$t_0,\;t$] every correct node will invoke at most a single pulse if $t-t_0\ge\cyclemin$
and will invoke at least a single pulse if $t-t_0\ge\cyclemax.$
\end{enumerate}
\end{definition}

The second Closure condition intends to tightly bound the effective
pulse invocation frequency within a priori bounds. This is in order
to defy any trivial solution that could synchronize the nodes, but
be completely unusable, such as instructing the nodes to invoke a
pulse every $\sigma$ time units. Note that this is a stronger
requirement than the ``linear envelope progression rate'' typically
required by clock synchronization algorithms, in which it is only
required that clock time progress as a linear function of real-time.

\section{The ``Pulse Synchronization'' Algorithm}
\label{sec:s3} \ignore{-2mm} We now present the {\sc
bio-pulse-synch} algorithm that solves the ``Pulse Synchronization''
problem defined in Definition~\ref{def:pulse-synch}, inspired by and
following a neurobiological analog. The \textbf{refractory function}
describes the time dependency of the firing threshold. At threshold
level $0$ the node invokes a pulse (\emph{fires}) {\bf
endogenously}. The algorithm uses several sub-procedures. With the
help of the {\sc \textbf{summation}} procedure, each node sums the
pulses that it learns about during a recent time window. If this sum
(called the {\em Counter}) crosses the current (time-dependent)
threshold for firing, then the node will fire, i.e broadcasts its
Counter value at the firing time. The exact properties of the time
window for summing messages is determined by the message decay time
in the {\sc prune} procedure (see Fig.~\ref{alg:prune}).

We now show in greater detail the elements and procedures described
above.

\vspace{+2mm} \noindent{\bf The refractory function }\\ The \Cycle\
is the predefined time a correct node will count on its timer before
invoking an endogenous pulse. The refractory function,
$REF(t):t\rightarrow \{0..n\!\!+\!\!1\},$ determines at every moment
the threshold for invoking a new pulse. The refractory function is
determined by the parameters \Cycle\, $n$, $f$, $d$ and $\rho.$ All
correct nodes execute the same protocol with the same parameters and
have the same refractory function. The refractory function is shaped
as a monotonously decreasing step function comprised of $n+2$ steps,
$REF\equiv (R_{n+1}, R_n, ..., R_0),$ where step $R_i \in
\mathbb{R}^+$ is the time length on the node's timer of threshold
level $i.$ The refractory function $REF,$ starts at threshold level
$n+1$ and decreases with time towards threshold level $0.$ The time
length of each threshold step is formulated in Eq.~\ref{eq:REF}:

\begin{equation} R_i=\left\{\begin{array}{ll}
\frac{\frac{1}{1-\rho}Cycle}{n-f}&i=1\ldots n-f-1\\\\
\frac{R_1-R_{n+1}-\frac{\rho}{1-\rho}Cycle}{f+1}&i=n-f\ldots n\\\\
2d(1+\rho)\cdot\frac{(\frac{1+\rho}{1-\rho})^{n+3}-1}{(\frac{1+\rho}{1-\rho})-1}&i=n+1,
\end{array}
\right. \label{eq:REF}
\end{equation}

Subsequent to a pulse invocation the refractory function is
restarted at $REF=n+1.$ The node will then commence threshold level
$n$ only after measuring $R_{n+1}$ time units on its timer.
Threshold level 0 ($REF=0$) is reached only if exactly \Cycle\ time
units have elapsed on a node's timer since the last pulse
invocation, following which threshold level $n+1$ is reached
immediately. Hence, by definition, $ \sum_{i=1}^{n+1}R_i \equiv
\Cycle.$ It is proven later in Lemma~\ref{lem:REF_prop_6} that $REF$
in Eq.~\ref{eq:REF} is consistent with this.

The special step $R_{n+1}$ is called the {\bf absolute refractory
period} of the cycle. Following the neurobiological analogue with
the same name, this is the first period after a node fires, during
which its threshold level is in practice ``infinitely high''; thus a
node can never fire within its absolute refractory period.

See Fig.~\ref{fig:figure1} for a graphical presentation of the
refractory function and its role in the main algorithm.

\vspace{+2mm} \noindent{\bf The message sent when firing} \\
The content of a message $M_p$ sent by a node $p,$ is the Counter,
which represents the number of messages received within a certain
time window (whose exact properties are described in the appendix)
that triggered $p$ to fire. We use the notation $Counter_p$ to mark
the local Counter at node $p$ and $Counter_{M_p}$ to mark the
Counter contained in a received message $M_p$ sent by node $p.$

\subsection{The {\sc SUMMATION} procedure}
\label{ssec:summation-sect}

A full account of the proof of correctness of the {\sc summation}
procedure is provided in the appendix. The {\sc summation} procedure
is executed upon the arrival of a new message. Its purpose is to
decide whether this message is eligible for being counted. It is
comprised of the following sub-procedures:

Upon arrival of the new message, the {\sc \textbf{timeliness}}
procedure determines if the Counter contained in the message seems
``plausible'' (\emph{timely}) with respect to the number of other
messages received recently (it also waits a short time for such
messages to possibly arrive). The bound on message transmission and
processing time among correct nodes allows a node to estimate
whether the content of a message it receives is plausible and
therefore timely. For example, it does not make sense to consider an
arrived message that states that it was sent as a result of
receiving $2f$ messages, if less than $f$ messages have been
received during a recent time window. Such a message is clearly seen
as a faulty node by all correct nodes. On the other, a message that
states that it was sent as a result of receiving $2f$ messages, when
$2f-1$ messages have been received during a recent time window does
not bear enough information to decide whether it is faulty or not,
as other correct nodes may have decided that this message is timely,
due to receiving a faulty message. Such a message needs to be
temporarily tabled so that it can be reconsidered for being counted
in case some correct node sends a message within a short time, and
which has counted that faulty message. Thus, intuitively, a message
will be timely if the Counter in that messages is less or equal to
the total number of tabled or timely messages that were received
within a short recent time window. The exact length of the
``recent'' time window is a crucial factor in the algorithm. There
is no fixed time after which a message is too old to be timely. The
time for message exchange between correct nodes is never delayed
beyond the network and processing delay. Thus, the fire of a correct
node, as a consequence of a message that it received, adds a bounded
amount of relay time. This is the basis for the time window within
which a specific Counter of a message is checked for plausibility.
Hence, a particular Counter of a message is plausible only if there
is a sufficient number of other messages (tabled or not) that were
received within a sufficiently small time window to have been
relayed from one to the other within the bound on relaying between
correct nodes. As an example, consider that the bound on the allowed
relay interval of messages is taken to be $2d$ time units. Suppose
that a correct node receives a message with Counter that equals $k.$
That message will only be considered as timely if there are at least
$k+1$ messages that were received (including the last one) in the
last $k\cdot2d$ time window. This is the main criterion for being
timely. On termination of the procedure the message is said to have
been \textbf{assessed}.

If a message is assessed as timely then the {\sc
\textbf{make-accountable}} procedure determines by how much to
increment the Counter. It does so by considering the minimal number
of recently tabled messages that were needed in order to assess the
message as timely. This number is the amount by which the Counter is
incremented by. A tabled message is marked as ``uncounted'' because
the node's Counter does not reflect this message. Tabled messages
that are used for assessing a message as timely become marked as
``counted'' because the node's Counter now reflect these message as
if they were initially timely. A node's Counter at every moment is
exactly the number of messages that are marked as ``counted'' at
that moment.

The {\sc \textbf{prune}} procedure is responsible for the tabling of
messages. A correct node wishes to  mark as counted, only those
messages which considering the elapsed time since their arrival,
will together pass the criterion for being timely at any correct
nodes receiving the consequent Counter to be sent. Thus, messages
that were initially assessed as timely are tabled after a short
while. This is what causes the Counter to dissipate. After a certain
time messages are deleted altogether ({\em decayed}).

\begin{figure}[h!]
\center \fbox{\begin{minipage}{4.6in} \footnotesize
\begin{alltt}
\setlength{\baselineskip}{4mm}

{\bf {\sc SUMMATION}({\em a new message $M_p$ arrived at time $t_{\mathrm arr}$})}\hfill/* {\em at node q} */\\

if ({\sc timeliness}$(M_p, t_{\mathrm arr})$ == ``$M_p$ is timely'') then \\
\tb {\sc make-accountable}$(M_p)$;\hfill /* {\em possibly increment $Counter_q$} */\\
\tb {\sc prune}$(t)$;

\end{alltt}
\normalsize
\end{minipage} }
\caption{The {\sc summation} procedure} \label{alg:summation-alg}
\end{figure}

The target of the {\sc summation} procedure is formulated in the
following two properties:

\vspace{2mm}\noindent{\bf Summation Properties:} Following the
arrival of a message from a correct node:
\begin{enumerate}

\item [\textbf{P1:}] The message is assessed within $d$ real-time units.

\item [\textbf{P2:}] Following assessment of the message the receiving node's Counter is incremented to hold a
value greater than the Counter in the message.
\end{enumerate}

\noindent The {\sc summation} procedure satisfies the Summation
Properties by the following heuristics:
\begin{itemize}
\item When the Counter crosses the threshold level, either
due to a sufficient counter increment or a threshold decrement, then
the node sends a message (fires). The message sent holds the value
of Counter at sending time.

\item The {\sc timeliness} procedure is employed at the receiving
node to assess the credibility (timeliness) of the value of the
Counter contained in this message. This procedure ensures that
messages sent by correct nodes with Counter less than $n$ will
always be assessed as timely by other correct nodes receiving this
message.

\item When a received message is declared timely and therefore
accounted for it is stored in a ``counted'' message buffer
(``Counted Set''). The receiving node's Counter is then updated to
hold a value greater than the Counter in the message by the {\sc
make-accountable} procedure.

\item If a message received is declared untimely then it is
temporarily stored in an ``uncounted'' message buffer (``Uncounted
Set'') and will not be accounted for at this stage. Over time, the
timeliness test of previously stored timely messages may not hold
any more. In this case, such messages will be moved from the Counted
Set to the Uncounted Set by the {\sc prune} procedure.

\item All messages are deleted after a certain time-period (message decay time) by the
{\sc prune} procedure.
\end{itemize}

\vspace{2mm}\noindent{\bf Definitions and state variables:}\\

\ignore{2mm}\noindent{\bf Counter}: an integer representing the
node's estimation of the number of timely firing events received
from distinct nodes within a certain time window. Counter is updated
upon receiving a timely message. The node's Counter is checked
against the refractory function whenever one of them changes. The
value of Counter is bounded and changes non-monotonously; the
arrival of timely events may increase it and the decay/untimeliness
of old events may decrease it.\\

\ignore{2mm}\noindent{\bf Stored message}: is a basic data structure
represented as ($S_p, t_{arr}$) and created upon arrival of a
message $M_p.$ $S_p$ is the $id$ (or signature) of the sending node
$p$ and $t_{arr}$ is the local arrival time of the message. We say
that two stored messages, ($S_p, t_1$) and ($S_q, t_2$), are
{\bf distinct} if $p\neq q.$\\

\ignore{2mm}\noindent{\bf Counted Set (CS)}: is a set of distinct
stored messages that determine the current value of Counter. The
Counter reflects the number of stored messages in the Counted Set. A
stored message is {\bf accounted for} in Counter, if it
was in CS when the current value of Counter was determined.\\

\ignore{2mm}\noindent{\bf Uncounted Set (UCS)}: is a set of stored
messages, not necessarily distinct, that have not been accounted for
in the current value of Counter and that are not yet due to decay. A
stored message is placed (tabled) in the UCS when its message
clearly reflects a faulty sending node (such as when multiple
messages from the same node are received) or because it is
not timely anymore.\\

\ignore{2mm}\noindent{\bf Retired UCS (RUCS)}: is a set of distinct
stored messages not accounted for in the current value of Counter
due to the elapsed local time since their arrival. These stored messages are awaiting deletion (decaying).\\

The CS and UCS are mutually exclusive and together reflect the
 messages received from other nodes in the preceding time
window. Their union is denoted the node's {\bf \MessagePool. }\\

\noindent \mbox{\boldmath $t_{{\mathrm send}\,M_p}$}: denotes the
local-time at which a node $p$ sent a message $M_p.$ An equivalent
definition of $t_{{\mathrm send}\,M_p}$ is the local-time at which a
receiving node $p$ is ready to assess whether to send a message
consequent to the arrival and processing of some other message.\\

\noindent\mbox{\boldmath $MessageAge(t, q, p)$}: is the elapsed
time, at time $t,$ on a node $q$'s clock since the most recent
arrival of a message from node $p, $ which arrived at local-time
$t_{arr}.$ Thus, its value at
 node $q$ at current local-time $t$ is given by $t-t_{arr},$ where
$M_p$ is the most recent message that arrived from $p.$ If no stored
message is held at $q$ for $p$ then
$MessageAge(t, q, p)=\infty.$\\

\noindent\mbox{\boldmath $CSAge(t)$}: denotes, at local-time $t,$
the
largest $MessageAge(t, q, \ldots)$ among the stored messages in CS of node $q$.\\

\noindent \mbox{\boldmath $\tau$}: denotes the function
$\tau(k)\equiv
2d(1+\rho)\frac{(\frac{1+\rho}{1-\rho})^{k+1}-1}{(\frac{1+\rho}{1-\rho})-1}\;\;
.$\\

%\newpage

\vspace{2mm}\noindent{\bf The set of procedures used by the {\sc
summation} procedure (at node $q$):}

\begin{figure}[!h]
\center \fbox{\begin{minipage}{5in} \footnotesize
\begin{alltt}
\setlength{\baselineskip}{3.5mm}

The following procedure moves and deletes obsolete stored messages.
It prunes the CS to hold only stored messages such that a message
sent holding the resultant Counter will be
assessed as timely at any correct node receiving the message.\\

{\bf {\sc PRUNE} }($t$) \hfill \textit{/* at node q \textit{\ }*/}

\vspace{-2mm}\begin{itemize} \ignore{-3mm}\item Delete from RUCS all
entries ($S_p, t$) whose $MessageAge(t, q, p) > \tau (n+2);$
\ignore{-7mm}\item Move to RUCS, from the $\MessagePool,$ all stored
messages ($S_p, t$) whose $MessageAge(t, q, p) > \tau(n+1);$
\ignore{-3mm}\item Move to UCS, from CS, stored messages, beginning
with the oldest, until: $CSAge(t) \le \tau(k-1),$ where $k=\max[1,\;
\|CS\|];$ \ignore{-3mm}\item Set $Counter := \|CS\|;$
\end{itemize}
\end{alltt}
\normalsize
\end{minipage} }
\caption{The {\sc prune} procedure} \label{alg:prune}
\end{figure}

\begin{figure}[!h]
\center \fbox{\begin{minipage}{5in} \footnotesize
\begin{alltt}
\setlength{\baselineskip}{3.5mm}

We say that $M_p$ has been {\bf assessed} by $q,$ once the following
procedure is completed. A message $M_p,$ is timely at local time
$t_{arr}$ at node $q$ once it is declared timely by the procedure,
i.e. 1: whether the Counter in the message is within its valid
range; 2: whether the sending node has recently sent a message, in
which only the latest is considered; 3: whether enough messages have
been received
recently to support the credibility of the Counter in the message.\\

{\bf {\sc TIMELINESS} }$(M_p, t_{arr})$ \hfill \textit{/* at node q
\textit{\ }*/}\\
\textit{/* check if Counter is valid\hfill\textit{\ }*/}\\
\noindent\tb{\bf Timeliness Condition 1:}\\
\tb If ($0 \le Counter_{M_p} \le n-1$) Then \\
\due Create a new stored message ($S_p, t_{arr}$) and insert it into UCS;\\
\tb Else\\
\due return ``$M_p$ is not timely'';\\

\noindent\textit{/* if an older message from same node already
exists then must be a faulty node.
Delete all its entries but the latest.\hfill\textit{\ }*/}\\
\noindent\tb{\bf Timeliness Condition 2:}\\
\tb If ($\exists (S_p, t),$ s.t. $t\neq t_{arr},$ in $\MessagePool$
$\cup$ RUCS) Then\footnote{We assume no concomitant messages are
stamped with the exact same arrival times at a correct node.
We assume that one can uniquely identify messages.}\\
    \due delete from $\MessagePool$ all ($S_p, t'$), where $t^\prime\neq t_{arr};$\\
\due return ``$M_p$ is not timely'';\\

\noindent\textit{/* check if $Counter_{M_{p}}$ seems credible with respect to the $\MessagePool$\hfill\textit{\ }*/}\\
\noindent\tb{\bf Timeliness Condition 3:}\\
\tb Let $k$ denote $Counter_{M_{p}}.$\\
\tb If (at some local-time $t$ in the interval $[t_{arr},\; t_{arr}
+ d(1+\rho )]:$\\
\tb $\|\{(S_r,t^\prime)|(S_r,t^\prime) \in \MessagePool,
MessageAge(t, q, r) \le \tau(k+1)\}\|\ge k+1$) Then\footnote{We
assume
the implementation can assess these conditions within the time window.}\\
\due return ``$M_p$ is timely'';\\
\tb Else\\
\due return ``$M_p$ is not timely'';\\

\end{alltt}
\normalsize
\end{minipage} }
\caption{The {\sc timeliness} procedure} \label{alg:timeliness}
\end{figure}

\begin{figure}[!ht]
\center \fbox{\begin{minipage}{5in} \footnotesize
\begin{alltt}
\setlength{\baselineskip}{3.5mm}

This procedure moves stored messages from UCS into CS and updates
the value of Counter. This is done in case the arrival of a new
timely message $M_p,$ has made previously uncounted stored messages eligible for being counted.\\

{\bf {\sc MAKE-ACCOUNTABLE} }($M_p$) \hfill \textit{/* at node q
\textit{\ }*/}
\begin{itemize} \ignore{-3mm}\item Move the $\max[1,\;
(Counter_{M_p} - Counter_q + 1)]$ most recent distinct stored
messages from UCS to CS; \ignore{-7mm}\item Set $Counter := \|CS\|;$
\end{itemize}
\end{alltt}
\normalsize
\end{minipage} }
\caption{The {\sc make-accountable} procedure}
\label{alg:make-accountable}
\end{figure}

\begin{figure}[h!]
\center \fbox{\begin{minipage}{5in} \footnotesize
\begin{alltt}
\setlength{\baselineskip}{3.5mm}

This procedure causes the effective \cycle\ of the node to be reset,
meaning that the $REF$ function starts the cycle from the highest
threshold
level again and down to threshold level 0.\\

{\bf {\sc CYCLE-RESET} }() \hfill \textit{/* at node q \textit{\
}*/}
\begin{itemize} \ignore{-3mm}\item Restart $REF$ at
$REF:=n\!\!+\!\!1;$
\end{itemize}
\end{alltt}
\normalsize
\end{minipage} }
\caption{The {\sc cycle-reset} procedure} \label{alg:cyclereset}
\end{figure}

\newpage
\mbox{\ }
\newpage

We now cite the main theorems of the {\sc summation} procedure. The
proofs are given in the appendix.

\begin{theorem}\label{thm:timely} Any message, $M_p,$ sent by a correct
node $p$ will be assessed as timely by every correct node $q.$
\end{theorem}

\begin{lemma}\label{lem:counter_inc} Following the arrival of a timely
message $M_p,$ at a node $q,$ then at time $t_{\scriptsize{\mathrm
send}\,M_q},$ $Counter_q
> Counter_{M_p}.$
\end{lemma}

\begin{theorem}\label{thm:SUMM_R1} The {\sc summation} procedure satisfies
the Summation Properties.
\end{theorem}

\begin{proof}Let $p$ denote a correct node that sends $M_p.$
Theorem~\ref{thm:timely} ensures that $M_p$ is assessed as timely at
every correct node. Lemma~\ref{lem:counter_inc} ensures that the
value of $Counter$ will not decrease below $Counter_{M_p} + 1$ until
local-time $t_{{\mathrm send}\,M_p},$ thereby satisfying the
Summation Properties.
\end{proof}

\newpage
%\mbox{\ }

\subsection{The event driven ``pulse synchronization''
algorithm} \label{ssec:pulse-synch-alg}

\vspace{-1mm}
Fig.~\ref{alg:pulse-synch-alg} shows the main algorithm.
Fig.~\ref{fig:figure1} illustrates the mode of operation of the main
algorithm.
\vspace{-1mm}
\begin{figure}[h!]
\center \fbox{\begin{minipage}{4.65in} \footnotesize
\begin{alltt}
\setlength{\baselineskip}{4mm}

{\bf {\sc BIO-PULSE-SYNCH}($n, f, \Cycle$)} \hfill /* {\em at node q} */\\
\\$\bullet$ It is assumed that all the
parameters and variables are
verified to be within their range of validity.\\
$\bullet$ \emph{t} is the local-time at the moment of executing the respective statement.\\
\\{\bf if\;} ({\em a new message $M_p$ arrives at time $t_{\mathrm arr}$}) \textbf{then}\\
\tb {\sc summation}($(M_p, t_{\mathrm arr})$);  \\
\tb if ($Counter_q \ge REF(t)$) then\\
\due Broadcast $Counter_q$ to all nodes; \hfill /* {\em invocation of the Pulse} */\\
\due {\sc cycle-reset}();\\

{\bf if} ({\em change in threshold level according to} $REF$) \textbf{then}\\
\tb {\sc prune}($t$);\\
\tb if ($Counter_q \ge REF(t)$) then\\
\due Broadcast $Counter_q$ to all nodes; \hfill /* {\em invocation of the Pulse} */\\
\due {\sc cycle-reset}();\\

\end{alltt}
\normalsize
\end{minipage} }
\caption{The event driven {\sc bio-pulse-synch} algorithm }
\label{alg:pulse-synch-alg}
\end{figure}

\begin{figure}[h!]
\begin{center}
\includegraphics[scale=0.57,clip=true]{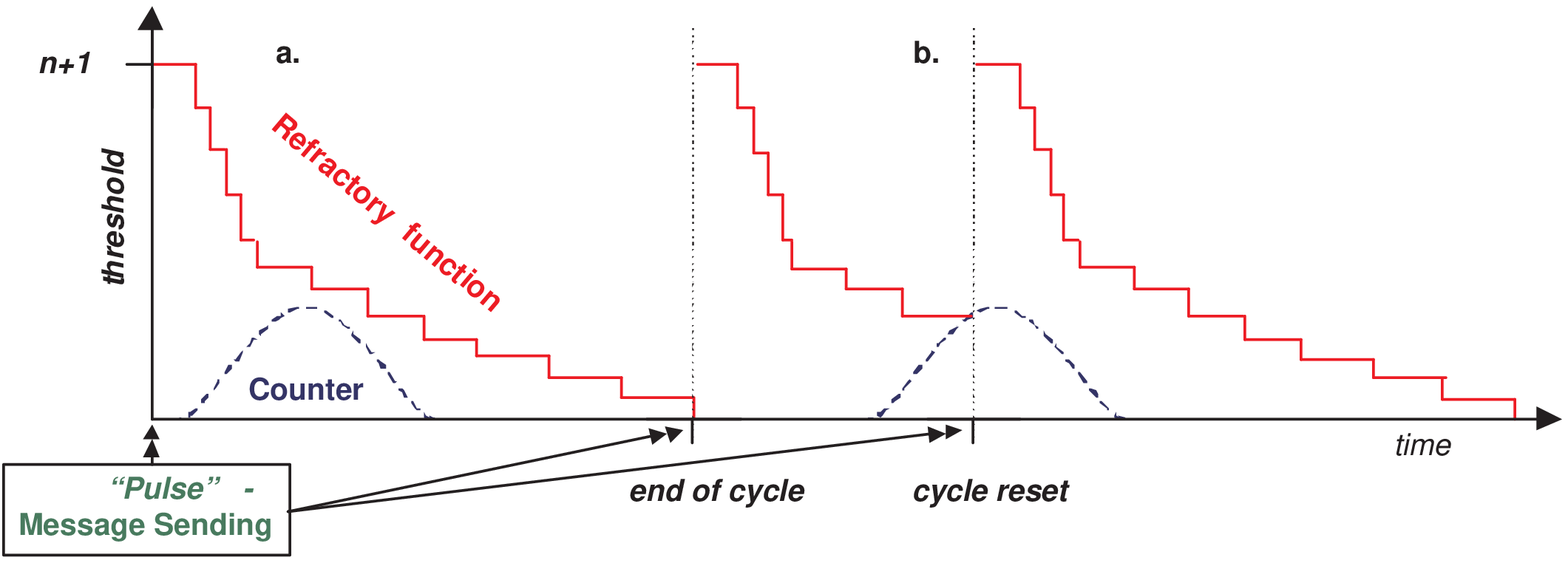}
\vspace{-1mm} \caption{\footnotesize Schematic example of the mode
of operation of {\sc bio-pulse-synch}: \textbf{(a.)} The node's
Counter (the summed messages) does not cross the threshold during
the cycle, letting the refractory function reach zero and
consequently the node fires endogenously. \textbf{(b.)} Sufficient
messages from other nodes are received in time window for the
Counter to surpass the current threshold, consequently the node
fires early and resets its cycle.} \label{fig:figure1} \vspace{-1mm}
\end{center}
\end{figure}

\newpage
\subsection{A Reliable-Broadcast Primitive}
\ignore{-2mm}\label{sec:IAGree}

In the current subsection we show that the {\sc BIO-PULSE-SYNCH}
algorithm can also operate in networks in which
Byzantine nodes may exhibit true two-faced behavior. This is done by
executing in the background a self-stabilizing Byzantine
reliable-broadcast-like primitive, which assumes no synchronicity
whatsoever among the nodes. It has the property of relaying any
message received by a correct node. Hence, this primitive satisfies
the broadcast assumption of Definition~\ref{def:net_nonfaulty} by
supplying a property similar to the relay property of the
reliable-broadcast primitive in \cite{FastAgree87}. That latter
primitive assumes a synchronous initialization and can thus not be
used as a building block for a self-stabilizing algorithm.

In \cite{DDSSBA-PODC06} we presented the \initiator primitive. We
say that a node does an \textbf{\ginit} of a message $m$ sent by
some node $p$ (denoted $\la p,m\ra$) if it accepts that this message
was sent by node $p.$

The \initiator primitive essentially satisfies the following two properties (rephrased for our purposes):\\

\begin{enumerate}

\item[\small\bf IA-1A] ({\em Correctness}) If all correct nodes
invoke \initiator$\la p,m\ra$ within $d$ real-time of each other
then all correct nodes \ginit $\la p,m\ra$ within $2d$ real-time
units of the time the last correct node invokes the primitive
\initiator$\la p,m\ra.$

\item[\small\bf IA-3A] ({\em Relay}) If a correct node $q$ \ginit s  $\la p,m\ra$ at real-time $t$,
 then every correct node $q^\prime$ \ginit s  $\la p,m\ra,$ at some real-time $t^\prime,$ with $|t-t^\prime|\le2 d.$

\end{enumerate}

The \initiator primitive requires a correct node not to send two
successive messages within less than $6d$ real-time of each other.
Following the {\sc BIO-PULSE-SYNCH} algorithm (see Timeliness
Condition 2, in the {\sc timeliness} procedure), non-faulty nodes
cannot fire more than once in every $2d(1+\rho)\cdot n > 6d$
real-time interval even if the system is not coherent, which thus
satisfies this requirement.

The use of the \initiator primitive in our algorithm is by executing
it in the background. When a correct node wishes to send a message
it does so through the primitive, which has certain conditions for
$\ginit$ing a message. Nodes may also \ginit\ messages that where
not sent or received through the primitive, if the conditions are
satisfied. In our algorithm nodes will deliver messages only after
they have been $\ginit$ed (also for the node's own message). From
[IA-1A] we get that all messages from correct nodes are delivered
within $3d$ real-time units subsequent to sending. From [IA-3A] we
have that all messages are delivered within $2d$ real-time units of
each other at all correct nodes, even if the sender is faulty. Thus,
we get that the new network delay $\tilde{d}=3d.$ Hence, the cost of
using the \initiator primitive is an added $2d$ real-time units to
the achieved pulse synchronization tightness which hence becomes
$\sigma=\tilde{d}=3d.$

\section{Proof of Correctness of {\sc BIO-PULSE-SYNCH}} \ignore{-2mm}\label{sec:s4}
In this section we prove Closure and Convergence of the {\sc
BIO-PULSE-SYNCH} algorithm. In the first subsection,
\ref{ssec:notations}, we present additional notations that
facilitate the proofs. In the second subsection, \ref{ssec:closure},
we prove Closure and in the third, \ref{ssec:convergence}, we prove
Convergence.

The proof that {\sc bio-pulse-synch} satisfies the pulse
synchronization problem follows the steps below:

Subsection~\ref{ssec:notations} introduces some notations and
procedures that are for proof purposes only. One such procedure
partitions the correct nodes into disjoint sets of synchronized
nodes (``synchronized clusters'').

In Subsection~\ref{ssec:closure} (Lemma~\ref{lem:l1}), we prove that
``synchronized clusters'' once formed stay as synchronized sets of
nodes, this implies that once the system is in a
synchronized\_pulse\_state it remains as such ({\em Closure}).

In Subsection~\ref{ssec:convergence}
(Theorem~\ref{thm:convergence}), we prove that within a finite
number of cycles, the synchronized clusters repeatedly absorb to
form ever larger synchronized sets of nodes, until a
synchronized\_pulse\_state of the system is reached ({\em
Convergence}).

Note that the the synchronization tightness, $\sigma,$ of our
algorithm, equals $d.$

It may ease following the proofs by thinking of the algorithm in the
terms of non-liner dynamics, though this is not necessary for the
understanding of any part of the protocol or its proofs. We show
that the state space can be divided into a small number of stable
fixed points (``synchronized sets'') such that the state of each
individual node is attracted to one of the stable fixed points. We
show that there are always at least two of these fixed points that
are situated in the basins of attraction (``absorbance distance'')
of each other. Following the dynamics of these attractors, we show
that eventually the states of all nodes settle in a limit cycle in
the basin of one attractor.

\subsection{Notations, procedures and properties used in the proofs}\label{ssec:notations}

\vspace{2mm}\noindent{\bf First node} in a synchronized set of nodes
$S,$ is a node of the subset of nodes that ``fire first'' in $S$
that satisfies: \ignore{-2mm}
\[ \small\mbox{``First node
in $S$''} = \left\{ \begin{array}{ll} \min \{ i | i\in \max \{
\phi_i(t) | \mbox{node } i\in  S, \phi_i(t) \le  \sigma \}\}
&\exists i\in
S \,\mbox{s.t.}\, \phi_i(t) \le \sigma\\

\min\{ i | i\in \max\{ \phi_i(t) | \mbox{node }i\in S \}\}  &
\mbox{\;otherwise.}
\end{array}
\right.
\] \ignore{-2mm}

Equivalently, we define \textbf{last node}: \ignore{-2mm} \[\small
\mbox{``Last node in $S$''} = \left\{\begin{array}{ll}

\max\{ i | i\in \min\{ \phi_i(t) | \mbox{node }i\in  S,
\phi_i(t) >\sigma \}\} &\exists i\in  S \,\mbox{s.t.}\, \phi_i(t) >\sigma \\

\max\{ i | i\in \min\{ \phi_i(t) | \mbox{node }i\in S \}\} &
\mbox{\;otherwise}.
\end{array}
\right.
\] \ignore{-2mm}

The second cases in both definitions serve to identify the First and
Last nodes in case $t$ falls in-between the fire of the nodes of the
set.

\vspace{+2mm} \noindent{\bf Synchronized Clusters}\\
At a given time $t$ the nodes are divided into disjoint {\bf
synchronized clusters} in the following way:
\begin{enumerate}
\ignore{-2mm} \item Assign the maximal synchronized set of nodes at
time $t$ as a synchronized cluster. In case there are several
maximal sets choose the set that is harboring the first node of the
unified set of all these maximal sets.

\ignore{-3mm} \item Assign the second maximal synchronized set of
nodes that are not part of the first synchronized cluster as a
synchronized cluster.

\ignore{-3mm} \item Continue until all nodes are exclusively
assigned to a synchronized cluster.
\end{enumerate}

The synchronized cluster harboring the node with the largest
(necessarily finite) $\phi$ among all the nodes is designated $C_1.$
The rest of the synchronized clusters are enumerated inversely to
the $\phi$ of their first node, thus if there are $m$ synchronized
clusters then $C_m$ is the synchronized cluster whose first node has
the lowest $\phi$ (besides perhaps $C_1$). Note that at most one
synchronized cluster may have nodes whose  actual $\phi$ differences
is larger than $\sigma,$ as it can contain nodes that have just
fired and nodes just about to fire. The definition of $C_1$ implies
that at the time the nodes are partitioned into synchronized
clusters (time $t$ above) it may be the only synchronized cluster in
such a state.

The clustering is done only for illustrative purposes of the proof.
It does not actually affect the protocol or the behavior of the
nodes. In the proof we ``assign'' the nodes to synchronized clusters
at some time $t.$ From that time on we consider the synchronized
clusters as a constant partitioning of the nodes into disjoint
synchronized sets of nodes and we follow the dynamics of these sets.
Thus, once a node is exclusively assigned to some synchronized
cluster it will stay a member of that synchronized cluster. We aim
at showing that eventually all synchronized clusters become one
synchronized set of nodes. Once such a clustering is fixated we
ignore nodes that happen to fail and forthcoming recovering nodes.
Our proof is based on the observation that eventually we reach a
time window within which the permanent number of non-correct nodes
at every time is bounded by $f$ and during that window the whole
system converges.

\begin{observation}\label{cor:synch_clu} The synchronized clustering
procedure assigns every correct node to exactly one synchronized set
of nodes.
\end{observation}

\begin{observation}\label{cor:synch_clu2} Immediately following the synchronized clustering
procedure no two distinct synchronized clusters comprise one
synchronized set of nodes.
\end{observation}

We use the following definitions and notations:\\

\noindent$\bullet\ C_i\ -$   synchronized cluster number $i.$ \\
$\bullet\ n_i\ -$   cardinality of $C_i$ (i.e. number of correct
nodes associated with
synchronized cluster $C_i$).\\
$\bullet\ c\ -$ current number of synchronized clusters in the current state; $c\ge 1.$\\
$\bullet\ dist(a, b, t) \equiv |\phi_a(t) - \phi_b(t)|$ is the {\bf
distance} ($\phi$ difference) between nodes $a$ and $b$ at real-time $t.$\\
$\bullet\ \phi_{c_i}(t)\ -$ is the $ \phi(t)$ of the first node in
synchronized cluster $C_i.$ \\
 $\bullet\ dist(C_i, C_j, t) \equiv dist(\phi_{c_i}(t),\;\phi_{c_j}(t),\;
t)$ at real-time $t.$\\

If at real-time $t$ there exists no other synchronized cluster
$C_r,$ such that $\phi_{c_i}(t) \ge \phi_{c_r}(t) \ge
\phi_{c_j}(t),$ then we say that the synchronized clusters $C_i$ and
$C_j$ are {\bf adjacent} at real-time $t.$

\ignore{
\begin{corollary}\label{cor:c0} If we let ($C_i, C_{(i+1)\!\!\!\pmod{ c}}$)
denote adjacent synchronized clusters, then \ignore{-2mm}
\begin{equation} \sum_{i=1}^c
dist(C_i,C_{(i+1)\!\!\!\!\!\!\pmod{c}})\le (1+\rho)\,
\Cycle\enspace.\label{eq:e5}
\end{equation}\dd{it is not a corollary, maybe to write "it can be shown that"}
\end{corollary}
}%end ignore

We say that two synchronized clusters, $C_i$ and $C_j,$ have {\bf
absorbed} if their union comprises a synchronized set of nodes. If a
node in $C_j$ fires due to a message received from a node in $C_i,$
then, as will be shown in Lemma~\ref{lem:abs_lem}, the inevitable
result is that their two synchronized clusters absorb. The course of
action from the arrival of the message at a node in $C_j$ until
$C_j$ has absorbed with $C_i$ is referred to as the {\bf absorbance}
of $C_j$ by $C_i.$

We refer throughout the paper to the {\bf fire of a synchronized
cluster} instead of referring to the sum of the fires of the
individual nodes in the synchronized cluster. In Lemma~\ref{lem:l3}
we prove that these two notations are equivalent.\\

In Theorem~\ref{thm:abs_thm} we show that we can explicitly
determine a threshold value, $\bf{ad(C_i),}$ that has the property
that if for two synchronized clusters $C_i$ and $C_j,$ $dist(C_i,
C_j, t) \le ad(C_i)$ then $C_i$ {\bf absorbs} $C_j.$ We will call
that value the ``\textbf{absorbance distance}'' of $C_i.$

\begin{definition}\label{def:ad} The absorbance distance, $ad(C_i),$ of a synchronized cluster $C_i,$ is
$$ad(C_i) \equiv
\sum_{g=f+1}^{f+n_i}\!\!\!R_g$$ real-time units.
\end{definition}

\noindent \textbf{Properties used for the proofs
}\\

 We identify and prove several properties; one property of the
{\sc summation} procedure (Property 1) and several properties of
$REF$ (Properties 2-7). These are later used to prove the
correctness of the algorithm.

\vspace{2mm}{\bf Property 1:} See the Summation Properties in
Subsection~\ref{ssec:summation-sect}.

{\bf Property 2:} $R_i$ is a monotonic decreasing function of $
i,$ $R_i \ge  R_{i+1},$ for $ i=1 \ldots n-1.  $\\

{\bf  Property 3:} $R_i > 3d +
\frac{2\rho}{1-\rho^2}\sum_{j=1}^{n+1}R_j,$ for $i=1 \ldots n-f-1.$\\

{\bf Property 4:} $R_i > \sigma(1-\rho) +
\frac{2\rho}{1+\rho}\sum_{j=1}^{n+1}R_j,$ for $ i=1 \ldots n.$\\

{\bf Property 5:} $R_{n+1} \ge 2d(1+\rho)
\frac{(\frac{1+\rho}{1-\rho})^{n+3}-1}{(\frac{1+\rho}{1-\rho})-1}\;.$\\

{\bf Property 6:} $R_1+\cdots+R_{n+1} = \Cycle.$\\

Consider any clustering of $n-f$ correct nodes into $c > 1$
synchronized clusters, in which $j^\prime$ denotes the largest
synchronized. Thus $n_{j^\prime}$ is the number of nodes in the
largest synchronized cluster and is less or equal to $n-f-1.$ The
number of nodes in the second largest cluster is less or equal to $\lfloor \frac{(n-f)}{2}\rfloor.$\\

{\bf Property 7:} \begin{equation} \sum_{j=1,j\ne
j^\prime}^{c}\sum_{g=f+1}^{f+n_j}\!\!\!R_g+\sum_{g=1}^{n_{j^\prime}}\!\!R_g
\ge
 \frac{1}{1-\rho}\Cycle \enspace , \mbox{where }\sum_{j=1}^{c}n_j=n-f\enspace . \label{eq:prop7}
\end{equation}

We require the following restriction on the relationship between
$\Cycle,\;d,\;n$ and
$f$ in order to prove that Properties~3-4 hold:\\

{\bf Restriction 1:} \begin{equation}
\Cycle>d\cdot\frac{(1-\rho^2)[(1-\rho)
(f+1)+2(1+\rho)\cdot\frac{(\frac{1+\rho}{1-\rho})^{n+3}-1}{(\frac{1+\rho}{1-\rho})-1}]}{\frac{1-\rho}{n-f}-3\rho+\rho^2}\enspace
.\label{eq:f-vs-cycle}\end{equation}\\

We now prove that Properties 2-7 are properties of $REF$:

\begin{lemma}\label{lem:REF_props} Properties 2-5 are properties of $REF$ under Restriction 1.
\end{lemma}

\begin{proof}
The proof for Properties 2 and 5 follows immediately from the
definition of $REF$ in Eq.~\ref{eq:REF}.

Note that $R_i > R_j,$ for $1 \le i \le n-f-1$ and $n-f \le j \le
n.$ Moreover, for $\sigma=d,$ Property~4 is more restrictive than
Property~3. Hence, for showing that Properties~3 and 4 are
properties of $REF$ it is sufficient to show that $R_j$ (where $n-f
\le j \le n$) satisfies Property~4:

\begin{eqnarray}
\lefteqn{R_j=\frac{R_1-R_{n+1}-\frac{\rho}{1-\rho}Cycle}{f+1}>\sigma(1-\rho)
+
\frac{2\rho}{1+\rho}\sum_{j=1}^{n+1}R_j \Rightarrow}\nonumber\\
&&
\frac{\frac{1}{1-\rho}\Cycle}{n-f}-2d(1+\rho)\cdot\frac{(\frac{1+\rho}{1-\rho})^{n+3}-1}{(\frac{1+\rho}{1-\rho})-1}-\frac{\rho}{1-\rho}Cycle>[d(1-\rho)
+
\frac{2\rho}{1+\rho}\Cycle](f+1) \Rightarrow\nonumber\\
&& \frac{1}{1-\rho}\Cycle-\frac{\rho}{1-\rho}(n-f)\Cycle-
\frac{2\rho}{1+\rho}(n-f)\Cycle \nonumber\\&&\tre\due>[d(1-\rho)
(f+1)+2d(1+\rho)\cdot\frac{(\frac{1+\rho}{1-\rho})^{n+3}-1}{(\frac{1+\rho}{1-\rho})-1}](n-f) \Rightarrow\nonumber\\
&& [\frac{1-\rho(n-f)}{1-\rho}-
\frac{2\rho}{1+\rho}(n-f)]\Cycle\nonumber\\&&\tre\due>d[(1-\rho)
(f+1)+2(1+\rho)\cdot\frac{(\frac{1+\rho}{1-\rho})^{n+3}-1}{(\frac{1+\rho}{1-\rho})-1}](n-f) \Rightarrow\nonumber\\
&&
[\frac{(\frac{1}{n-f}-\rho)(1+\rho)-2\rho(1-\rho)}{1-\rho^2}]\Cycle>d[(1-\rho)
(f+1)+2(1+\rho)\cdot\frac{(\frac{1+\rho}{1-\rho})^{n+3}-1}{(\frac{1+\rho}{1-\rho})-1}] \Rightarrow\nonumber\\
&& \frac{\frac{1-\rho}{n-f}-3\rho+\rho^2}{1-\rho^2}\Cycle>d[(1-\rho)
(f+1)+2(1+\rho)\cdot\frac{(\frac{1+\rho}{1-\rho})^{n+3}-1}{(\frac{1+\rho}{1-\rho})-1}] \Rightarrow\nonumber\\
&& \Cycle>d\cdot\frac{(1-\rho^2)[(1-\rho)
(f+1)+2(1+\rho)\cdot\frac{(\frac{1+\rho}{1-\rho})^{n+3}-1}{(\frac{1+\rho}{1-\rho})-1}]}{\frac{1-\rho}{n-f}-3\rho+\rho^2}
\enspace . \label{eq:props2-5}
\end{eqnarray}

This inequality is exactly satisfied by Restriction 1 and thus
Eq.~\ref{eq:REF} satisfies Properties~3 and 4.

Note that for $\rho=0,$ the inequality becomes
$\Cycle>d\cdot(f+1)(n-f).$
\end{proof}

\begin{lemma}\label{lem:REF_prop_6} Property 6 is a property of $REF.$
\end{lemma}

\begin{proof}\begin{eqnarray*} \lefteqn{R_1+\cdots+R_{n+1} =
(R_1+\cdots+R_{n-f-1}) + (R_{n-f}+\cdots+R_n)
+ R_{n+1} }\\
& &=(n-f-1)\cdot \frac{\frac{1}{1-\rho}\Cycle}{n-f}+(f+1)\cdot \frac{R_1-R_{n+1}-\frac{\rho}{1-\rho}\Cycle}{f+1}+R_{n+1}\\
& & =
\frac{1}{1-\rho}\Cycle-\frac{\frac{1}{1-\rho}\Cycle}{n-f}+R_1-R_{n+1}-\frac{\rho}{1-\rho}\Cycle+R_{n+1}=\Cycle\enspace.
\end{eqnarray*}
\end{proof}

\begin{lemma}\label{lem:REF_prop_7} Property 7 is a property of $REF$.
\end{lemma}

\begin{proof} We will prove that the constraint in Eq.~\ref{eq:prop7}
 is always satisfied by the refractory function in Eq.~\ref{eq:REF}.

Note that Eq.~\ref{eq:prop7} is a linear equation of the $R_i$
values of $REF$. We denoted $n_{j^\prime}$ to be the number of nodes
in the largest synchronized cluster, following some partitioning of
the correct nodes into synchronized clusters. We want to find what
is the largest value of $i$ such that $R_i$ is a value with a
non-zero coefficient in the linear equation Eq.~\ref{eq:prop7}. This
value is determined by either the largest possible cluster, which
may be of size $n-f-1$ (in case all but one of the correct nodes are
in one synchronized cluster\footnote{The case in which the $n-f$
correct nodes are in one synchronized cluster implies the objective
has been reached.}), or by the second-largest possible cluster,
which may be of size $\lfloor \frac{(n-f)}{2}\rfloor$ (in case all
correct nodes are in two possibly equally sized synchronized
clusters). Thus the largest value of $i$ such that $R_i$ is a value
with a non-zero coefficient equals max[$f+\lfloor
\frac{(n-f)}{2}\rfloor,\; n-f-1]=n-f-1,$ for $n \ge 3f+1.$

Thus, following Eq.~\ref{eq:REF}, each of these $R_i$ values equals
$\frac{\frac{1}{1-\rho}\Cycle}{n-f}.$ There are exactly $n-f$ (not
necessarily different) $R_i$ values in Eq.~\ref{eq:prop7}. Hence,
incorporating Eq.~\ref{eq:REF} into Eq.~\ref{eq:prop7} reduces
Eq.~\ref{eq:prop7}
 to the linear equation:
$(n-f)\cdot R_i \ge \frac{1}{1-\rho}\Cycle,$ where $1\le i\le
n-f-1.$ It remains to show that Eq.~\ref{eq:REF} satisfies this
constraint:
$$(n-f)\cdot R_i = (n-f)\cdot \frac{\frac{1}{1-\rho}\Cycle}{n-f} = \frac{1}{1-\rho}\Cycle.$$

\end{proof}

\subsection{Proving the Closure} \label{ssec:closure}\ignore{-2mm}

We now show that a synchronized set of nodes stays synchronized.
This also implies that the constituent nodes of a synchronized
clusters stay as a synchronized set of nodes, as a synchronized
cluster is in particular a synchronized set of nodes. This proves
the first Closure requirement of the ``Pulse Synchronization''
problem in Definition~\ref{def:pulse-synch}.

\ignore{
\begin{figure}[tb]
\ignore{-10mm}
\begin{minipage}{122mm}
%\hspace{2mm}
\centerline{\psfig{figure=figure3.eps}} \caption{\footnotesize A
schematic drawing of the scenario in which Lemma~\ref{lem:l1} is
proven. $\overline{\mbox{\ \ \ \ }}$ denotes the refractory function
of the first node. -\ -\ - denotes the refractory function of the
last node at a shift of $d=\sigma.$ The event of Counter=$k$ at the
first node is symbolized by the solid up arrow. The last node sets
its Counter to $k+1$ at the time pointed to by the dashed arrow, at
its threshold level $ k+1.$ $R_{k+1} $  is the time length of the
last node's threshold level $k+1.$} \label{fig:figure3}
\end{minipage}
\end{figure}
}

\begin{lemma}\label{lem:l1} A set of correct nodes that is a synchronized
set at real-time $t',$ remains synchronized $\forall t,\;t \ge t'.$
\end{lemma}

\begin{proof}
Let there be a synchronized set of nodes at real-time $t'.$ From the
definition of a synchronized set of nodes, this set of nodes will
stay synchronized as long as no node in the set fires. This is
because the $\phi$ difference between nodes (in real-time units)
does not change as long as none of them fires. We therefore turn our
attention to the first occasion after $t'$ at which a node from the
set fires. Let us examine the extreme case \ignore{(illustrated in
Fig.~\ref{fig:figure3})} of a synchronized set consisting of at
least two nodes at the maximal allowed $\phi$ difference; that is to
say that at time $t',$ $dist(first\_node, last\_node, t') = \sigma.$
Further assume that the first node in the set fires with a
Counter=$k,$ $(0\le k \le n-1),$ at some time $t \ge t'$ at the very
beginning of its threshold level $k,$ and without loss of generality
is also the first node in the set to fire after time $t'.$ We will
show that the rest of the nodes in the set will fire within the
interval $[t, t+\sigma]$ and thus remains a synchronized set.

Property~1 ensures that the last node's Counter will read at least
$k+1$ subsequently to the arrival and assessment of the first node's
fire, since its Counter should be at least the first node's Counter
plus $1.$ The proof of the lemma will be done by showing that right
after the assessment of the first node's fire, the last node cannot
be at a threshold higher than $k+1$ and thus will necessarily fire.

The proof is divided into the following steps:
\begin{enumerate}
\item Show that when the first node is at threshold level $k$ then
the last node is at threshold level $k+1$ or lower.

\item Show that
if the first node fires with a Counter=$k$ then due to Property 1
and Step 1 the last node will fire consequently.

\item Show that the last node
fires within a $d$ real-time window of the first node, and as a
result, the new distance between the first and last node is less
than or equal to $\sigma.$
\end{enumerate}

Observe that the extreme case considered is a worst case since if
the largest $\phi$ difference in the set is less than $\sigma$ then
the threshold level of the last node may only be lower. The same
argument also holds if the first node fires after its beginning of
its threshold level $k.$ Thus the steps of the proofs also apply to
any intermediate node in the synchronized set and thus remains as a
synchronized set of nodes.

\vspace{+2mm}\noindent{\bf Step 1:} In this step we aim at
calculating the amount of time on the last node's clock remaining
until it commences its threshold level $k,$ counting from the event
in which the first nodes commences its threshold level $k.$ By
showing that this remaining time is less than the length of
threshold level $k+1,$ as counted on the clock of the last and
slowest node we conclude that this node must be at most at threshold
level $k+1.$ The calculations are done on the slow node's clock.

Assume the first node to be the fastest permissible node and the
last one the slowest. Hence, when the first node's threshold level
$k$ commences,
\begin{equation}
\frac{1}{1+\rho}\sum_{i=k+1}^{n+1}R_i \label{eq:e6}\end{equation}
real-time units actually passed since it last fired. The last node
``counted'' this period as:
\begin{equation}
\frac{1-\rho}{1+\rho}\sum_{i=k+1}^{n+1}R_i\enspace .
\label{eq:e7}\end{equation}

The last node has to count on its clock, from the time that the
first node fired, at most $\sigma(1-\rho )$ local-time units (max.
$\phi$ difference of correct nodes in a synchronized set as counted
by the slowest node), and

\begin{equation}
\sum_{i=k+1}^{n+1}R_i \label{eq:slowk}
\end{equation}

 in order to reach its own threshold
level $k.$ As a result, the maximum local-time difference between
the time the first node starts its threshold level $k$ till the last
node starts its own threshold level $k$ as counted by the last node
is therefore $\sigma(1+\rho )$ plus the difference
Eq.~\ref{eq:slowk} -- Eq.~\ref{eq:e7}, which yields

\begin{equation}
\sigma(1-\rho)+\frac{1+\rho}{1+\rho}\sum_{i=k+1}^{n+1}R_i -
\frac{1-\rho}{1+\rho}\sum_{i=k+1}^{n+1}R_i =
\sigma(1-\rho)+\frac{2\rho}{1+\rho}\sum_{i=k+1}^{n+1}R_i \enspace
 .\label{eq:e8}\end{equation}

Property~4 ensures that $R_{k+1}$ is greater than Eq.~\ref{eq:e8}
for $0 \le k \le n-1;$ thus when the first node commences threshold
level $k$ the last node must be at a threshold level that is less or
equal to $k+1.$

\vspace{+2mm}\noindent{\bf Step 2:} Let the first node fire as a
result of its Counter equalling $k$ at time $t$ at threshold level
$k.$ In case that the last node receives almost immediately the
first node's fire (and thus increments its Counter to at least $k+1$
following Property~1), it must be at a threshold level that is less
or equal to $k+1$ (following Step~1) and will therefore fire. All
the more so if the first node's fire is received later, since the
threshold level can only decrease in time before a node fires.

\vspace{+2mm}\noindent{\bf Step 3:} We now need to estimate the new
distance between the first and last node in order to show that they
 still comprise a synchronized set. The last node assesses the first node's fire
within $d$ real-time units after the first node sent its message
(per definition of $d$). This yields a distance of $d(1-\rho )$ as
seen by the last node, which equals the maximal allowed real-time
distance, $d\ (= \sigma),$ between correct nodes in a synchronized
set at real-time $t',$ and thus they stay a synchronized set at
time~$t'.$  \end{proof}

\begin{corollary}\label{cor:closure1} (Closure 1) Lemma~\ref{lem:l1}
implies the first Closure condition.
\end{corollary}

\begin{lemma}\label{lem:closure2} (Closure 2) As long as the system state is in a\\ synchronized\_pulse\_state
then the second Closure condition holds.
\end{lemma}

\begin{proof} Due to
Lemma~\ref{lem:l1} the first node to fire in the synchronized set
following its previous pulse, may do so only if it receives the fire
from faulty nodes or if it fires endogenously. This may happen the
earliest if it receives the fire from exactly $f$ distinct faulty
nodes. Thus following Eq.~\ref{eq:REF} its cycle might have been
shortened by at most $f\cdot\frac{\Cycle}{n-f}$ real-time units.
Hence, in case the first node to fire is also a fast node, it
follows that $\cycle_{\mbox{\scriptsize\slshape
min}}=\Cycle\cdot(1-\rho)-\frac{f}{n-f}\cdot\Cycle\cdot(1-\rho)=\frac{n-2f}{n-f}\cdot\Cycle\cdot(1-\rho)$
real-time units. A node may fire at the latest if it fires
endogenously. If in addition it is a slow node then it follows that
$\cyclemax= \Cycle\cdot(1+\rho)$ real-time units.

Thus in any real-time interval that is less or equal to $\cyclemin$
any correct node will fire at most once. In any real-time interval
that is greater or equal to $\cyclemax$ any correct node will fire
at least once. This concludes the second closure condition.
\end{proof}

\subsection{Proving the Convergence} \label{ssec:convergence}\ignore{-2mm}
The proof of Convergence is done through several lemmata. We begin
by presenting sufficient conditions for two synchronized clusters to
absorb. In Subsection~\ref{ssec:s4.2}, we show that the refractory
function $REF$ ensures the continuous existence of a pair of
synchronized clusters whose unified set of nodes is not
synchronized, but are within an absorbance distance and hence
absorb. Thus, iteratively, all synchronized clusters will eventually
absorb to form a unified synchronized set of nodes.\\

% begin ignore
\ignore{

\begin{figure}[tb]
\ignore{-10mm}
\begin{minipage}{122mm}
%\hspace{2mm}
\centerline{\psfig{figure=figure4.eps}} \caption{\footnotesize A
schematic drawing of the rationale behind Lemma~\ref{lem:abs_lem}. -
- - - denotes the refractory function of the first node of $C_i.$ -\
-\ -\ - denotes the refractory function of the first node of $C_j.$
The event of Counter=$k$ at the first node of $C_i$ at the beginning
of threshold level $k$ is symbolized by the solid up-arrow. The
first node of $C_j$ receives the fire and increments its Counter to
$k+n_i$ at the time pointed to by the dashed up-arrow, at the
beginning of threshold level $k+n_i.$} \label{fig:figure4}
\end{minipage}
\end{figure}

} % end ignore

\begin{lemma}\label{lem:abs_lem} (Conditions for Absorbance) Given two synchronized clusters,
$C_i$ preceding $C_j,$ if: \begin{enumerate} \item $C_i$ fires with
Counter=$k,$ at real-time $t_{c_i\_fires},$ where $0\le k\le f$
\item $dist(C_i, C_j, t_{c_i\_fires}) \le
\frac{1}{1-\rho}\, \sum_{g=k+1}^{k+n_i}R_g -
\frac{2\rho}{1-\rho^2}\,\sum_{g=k+1}^{n+1}R_g$ \end{enumerate} then
$C_i$ will absorb $C_j .$
\end{lemma}

\begin{proof}
The proof is divided into the following steps:

\begin{enumerate}
\item
\begin{enumerate}
\item If $C_i$ fires before $C_j,$ then $C_j$ consequently fires.

\item Subsequent to the previous step: $dist(C_i, C_j,\;..) \le 3d.$

\end{enumerate}

\item Following the previous step, within one cycle the
constituent nodes of the two synchronized clusters comprise a
synchronized set of nodes.

\end{enumerate}

\noindent\textbf{Step 1a:} Let us examine the case in which $C_i$
fires first at some real-time denoted $t_{c_i\_fires},$ and in the
worst case that $C_j$ doesn't fire before it receives all of $C_i$'s
fire. All the calculations assume that at $t_{c_i\_fires},$
$\phi_{c_i}(t_{c_i\_fires})$ has still not been reset to 0.
Specifically, assume that the first node in $C_i$ fired due to
incrementing its Counter to $k$ $(0\le k\le f)$ at the beginning of
its threshold level $k.$ Following Property 1 and Lemma~\ref{lem:l3}
the nodes of $C_j$ increment their Counters to $k+n_i$ after
receiving the fire of $C_i.$ Additionally, in the worst case, assume
that the first node in $C_j$ receives the fire of $C_i$ almost
immediately. We will now show that this fire is received at a
threshold level $\le k+ n_i.$

We will calculate the upper-bound on the $\phi$ of the first node in
$C_j$ at real-time $t_{c_i\_fires},$ and hence deduce the
upper-bound on its threshold level. Assume the nodes of $C_i$ are
fast and the nodes of $C_j$ are slow. Should the nodes of $C_j$ be
faster, then the threshold level may only be lower.

\vspace{-2mm}\begin{eqnarray} \lefteqn{
\phi_{c_j}(t_{c_i\_fires})=}\nonumber\\&& =
\phi_{c_i}(t_{c_i\_fires})-[\frac{1}{1-\rho}\,
\sum_{g=k+1}^{k+n_i}R_g -
\frac{2\rho}{1-\rho^2}\,\sum_{g=k+1}^{n+1}R_g] \nonumber\\
 & & =
\frac{1}{1+\rho}\!\!\sum_{g=k+1}^{n+1}\!\!R_g\nonumber
-[\frac{1}{1-\rho}\sum_{g=k+1 }^{k+n_i }\!\!R_g-
\frac{2\rho}{1-\rho^2}\sum_{g=k+1}^{n+1}\!\!R_g]\enspace\\&& =
\frac{1}{1+\rho}\!\!\sum_{g=k+1}^{n+1}\!\!R_g\nonumber -[
\frac{1}{1-\rho}\sum_{g=k+1 }^{k+n_i }\!\!R_g+
(\frac{1}{1+\rho}-\frac{1}{1-\rho})\sum_{g=k+1}^{n+1}\!\!R_g]\enspace\\&&
=\frac{1}{1+\rho}\!\!\sum_{g=k+1}^{n+1}\!\!R_g\nonumber
-[(\frac{1}{1+\rho}-(\frac{1}{1+\rho}-\frac{1}{1-\rho}))\sum_{g=k+1
}^{k+n_i }\!\!R_g+
(\frac{1}{1+\rho}-\frac{1}{1-\rho})\sum_{g=k+1}^{n+1}\!\!R_g]\enspace\\&&
=\frac{1}{1+\rho}\!\!\sum_{g=k+1}^{n+1}\!\!R_g\nonumber
-[\frac{1}{1+\rho}\!\!\sum_{g=k+1 }^{k+n_i }\!\!R_g+
(\frac{1}{1+\rho}-\frac{1}{1-\rho})\!\!\sum_{g=k+1+n_i }^{n+1
}\!\!\!\!\!R_g]\enspace\\&&
=\frac{1}{1+\rho}\!\!\sum_{g=k+1}^{n+1}\!\!R_g\nonumber
-[\frac{1}{1+\rho}\!\!\sum_{g=k+1}^{n+1}\!\!R_g -
\frac{1}{1-\rho}\!\!\sum_{g=k+1+n_i}^{n+1}\!\!\!\!\!R_g]\enspace\\&&
=\frac{1}{1-\rho}\!\!\sum_{g=k+1+n_i}^{n+1}\!\!\!\!\!R_g\;.\nonumber\enspace\\&&
\label{eq:phase_cj}
\end{eqnarray}

We now seek to deduce the bound on $C_j$'s threshold level at the
time of $C_i$'s fire. Thus, following Eq.~\ref{eq:phase_cj}, at
real-time $t_{c_i\_fires}$ the $\phi$ of the first node in $C_j$ is
at most $\frac{1}{1-\rho}\!\sum_{g=k+1+n_i}^{n+1}\!R_g.$ We assumed
the worst case in which the constituent correct nodes of $C_j$ are
slow, thus these nodes have counted on their timers at least
$(1-\rho)\cdot\frac{1}{1-\rho}\!\sum_{g=k+1+n_i}^{n+1}\!R_g=\sum_{g=k+1+n_i}^{n+1}\!R_g$
time units since their last pulse. Hence, the correct nodes of $C_j$
are at real-time $t_{c_i\_fires}$ at most in threshold level
$k+n_i.$ Should $k<f$ or the fire of $C_i$ be received at a delay,
then this may only cause the threshold level at time of assessment
of the fire from $C_i$ to be equal or even smaller than $k+n_i.$
Thus, Lemma~\ref{lem:l1} and Property~1 guarantee that the first
node in $C_j$ will thus fire and that the rest of the nodes in both
synchronized clusters will follow their
respective first ones within $\sigma$ real-time units.\\

\noindent\textbf{Step 1b:} We seek to estimate the maximum distance
between the two synchronized clusters following the fire of $C_j.$
The first node in $C_j$ will fire at the latest upon receiving and
assessing the message of the last node in $C_i.$ More precisely,
fire at the latest $d$ real-time units following the fire of the
last node in $C_i,$ yielding a new $dist(C_i , C_j,\;..)$ of at most
$2d$ real-time units regardless of the previous $dist(C_i ,
C_j,\;..),$ $n_i,$ $k$ and $n_j.$ The last node of $C_j$ is at most
at a distance of $d$ from the first node of $C_j$ therefore making
the maximal distance between the first node of $C_i$ and the last
node of $C_j,$ at the
moment it fires, equal $3d$ real-time units.\\

\noindent\textbf{Step 2:} We will complete the proof by showing that
after $C_i$ causes $C_j$ to fire, the two synchronized clusters
actually absorb. We need to show that in the cycle subsequent to
Step 1, the nodes that constituted $C_i$ and $C_j$ become a
synchronized set. Examine the case in which following Step 1, either
one of the two synchronized clusters increment its Counter to
$k^\prime$ and fires at the beginning of threshold level $k^\prime.$
We will observe the $\phi$ of the first node to fire, denoted by
$\phi_{\mathrm first\_node-2nd-cycle}.$ Following the same arguments
as in Step 1, all other nodes increment their Counters to
$k^\prime+1$ after receiving this node's fire. Consider that this
happens at the moment that this first node incremented its Counter
to $k^\prime$ and fired, denoted $t_{2nd-cycle-fire}.$ Below we
compute, using Property~3, the lower bound on the $\phi$ of the rest
of the nodes at real-time $t_{2nd-cycle-fire},$ denoted
$\phi_{\mathrm other-nodes}(t_{2nd-cycle-fire}).$

\vspace{-2mm}\begin{eqnarray} \lefteqn{ \phi_{\mathrm
other-nodes}(t_{2nd-cycle-fire}) \ge
\phi_{\mathrm first\_node-2nd-cycle}(t_{2nd-cycle-fire})-3d} \nonumber\\
 & & =
\frac{1}{1+\rho}\!\!\sum_{g=k^\prime+1}^{n+1}\!\!R_g - 3d =
\frac{1}{1+\rho}\!\!\sum_{g=k^\prime+2}^{n+1}\!\!R_g\nonumber +
R_{k^\prime+1}-3d\enspace\\&& >
\frac{1}{1+\rho}\!\!\sum_{g=k^\prime+2}^{n+1}\!\!R_g +
\frac{2\rho}{1-\rho^2}\sum_{g=1}^{n+1}R_g\;.
 \label{eq:phase_2ndcycle}
\end{eqnarray}

In the worst case, the rest of the constituent nodes that were in
$C_i$ and $C_j$ are slow nodes and thus, at real-time
$t_{2nd-cycle-fire},$ counted:

\vspace{-2mm}\begin{eqnarray}
\lefteqn{(1-\rho)\cdot(\frac{1}{1+\rho}\!\!\sum_{g=k^\prime+2}^{n+1}\!\!R_g
+ \frac{2\rho}{1-\rho^2}\sum_{g=1}^{n+1}R_g) =
\frac{1-\rho}{1+\rho}\!\!\sum_{g=k^\prime+2}^{n+1}\!\!R_g +
\frac{2\rho}{1+\rho}\sum_{g=1}^{n+1}R_g}\nonumber\\
&&=
\frac{1-\rho}{1+\rho}\!\!\sum_{g=k^\prime+2}^{n+1}\!\!R_g +
\frac{2\rho}{1+\rho}\sum_{g=k^\prime+2}^{n+1}R_g+
\frac{2\rho}{1+\rho}\sum_{g=1}^{k^\prime+1}R_g\nonumber\\
&&= \sum_{g=k^\prime+2}^{n+1}\!\!R_g +
\frac{2\rho}{1+\rho}\sum_{g=1}^{k^\prime+1}R_g >
\sum_{g=k^\prime+2}^{n+1}\!\!R_g\enspace .
 \label{eq:phase_2ndcycle_counted}
\end{eqnarray}

time units since their last pulse. Due to Property~3 all these
correct nodes receive the fire and increment their Counters to
$k^\prime+1$ in a threshold level which is less or equal to
$k^\prime+1$ and will fire as well within $d$ real-time units of the
first node in the second cycle.
 \end{proof}

\begin{theorem}\label{thm:abs_thm} (Conditions for Absorbance) Given two synchronized clusters,
$C_i$ preceding $C_j,$ if: \begin{enumerate} \item $C_i$ fires with
Counter=$k,$ at real-time $t_{c_i\_fires},$ where $0\le k\le f,$ and
\item $\exists\;t,$  $t_{prev\_c_j\_fired} \le t \le t_{c_i\_fires},$
for which $dist(C_i, C_j, t) \le
ad(C_i)$ \end{enumerate} then $C_i$ will absorb $C_j.$
\end{theorem}

\begin{proof} Denote $t_{prev\_c_j\_fired}$ the real-time at which $C_j$
previously fired before time $t_{c_i\_fires}.$ Given that at some
time $t,$ where $t_{prev\_c_j\_fired} \le t\le t_{c_i\_fires},$
$dist(C_i,C_j, t) \le ad(C_i),$ we wish to calculate the maximal
possible distance between the two synchronized clusters at real-time
$t_{c_i\_fires},$ the time at which $C_i$ fires with
Counter=$k,$ where $0\le k\le f.$\\

Under the above assumptions, the maximal possible distance at
real-time $t_{c_i\_fires}$ is obtained when $k=f$ and when at time
$t_{prev\_c_j\_fired}$ the distance between $C_i$ and $C_j$  was
exactly $ad(C_i),$ i.e $dist(C_i, C_j,t_{prev\_c_j\_fired}) =
ad(C_i).$ The upper bound on $dist(C_i, C_j, t_{c_i\_fires})$ takes
into account that from $C_i's$ previous real-time firing time,
$t_{prev\_c_i\_fired},$ and until real-time $t_{c_i\_fires},$ the
nodes of $C_i$ were fast and that from real-time
$t_{prev\_c_j\_fired}$ and until $t_{c_i\_fires},$ the nodes of
$C_j$ were slow. Thus the bound on $dist(C_i ,C_j, t_{c_i\_fires})$
becomes the real-time difference between these:

\begin{eqnarray}
\lefteqn{dist(C_i ,C_j, t_{c_i\_fires}) = \phi_{c_i}(t_{c_i\_fires})
- \phi_{c_j}(t_{c_i\_fires})  =  }\nonumber\\
& & \frac{1}{1+\rho}\!\!\sum_{g=k+1}^{n+1}\!\!R_g -
\frac{1}{1-\rho}\!\!\sum_{g=k+1+n_i}^{n+1}\!\!\!\!\!R_g
=\frac{1}{1+\rho}\!\!\sum_{g=k+1 }^{k+n_i }\!\!R_g+
(\frac{1}{1+\rho}-\frac{1}{1-\rho})\!\!\sum_{g=k+1+n_i }^{n+1
}\!\!\!\!\!R_g=
\nonumber\\
&&(\frac{1}{1+\rho}-(\frac{1}{1+\rho}-\frac{1}{1-\rho}))\sum_{g=k+1
}^{k+n_i }\!\!R_g+
(\frac{1}{1+\rho}-\frac{1}{1-\rho})\sum_{g=k+1}^{n+1}\!\!R_g= \nonumber\\
&&
 \frac{1}{1-\rho}\sum_{g=k+1 }^{k+n_i }\!\!R_g+
(\frac{1}{1+\rho}-\frac{1}{1-\rho})\sum_{g=k+1}^{n+1}\!\!R_g=
\nonumber\\&& \frac{1}{1-\rho}\sum_{g=k+1 }^{k+n_i }\!\!R_g-
\frac{2\rho}{1-\rho^2}\sum_{g=k+1}^{n+1}\!\!R_g\enspace .
\label{eq:calc_ad}
\end{eqnarray}

Eq.~\ref{eq:calc_ad} is the upper bound on the distance between the
two synchronized clusters at real-time $t_{c_i\_fires},$ thus
following Lemma~\ref{lem:abs_lem}, the two synchronized clusters
absorb.
\end{proof}

\subsubsection{Convergence of the Synchronized Clusters}
\label{ssec:s4.2}\ignore{-2mm} In the coming subsection we look at
the correct nodes as partitioned into synchronized clusters (at some
specific time). Observation~\ref{cor:synch_clu2} ensures that no two
of these synchronized clusters comprise one synchronized set of
nodes. The objective of Theorem~\ref{thm:t2} is to show that within
finite time, at least two of these synchronized clusters will
comprise one synchronized set of nodes. Specifically, we show that
in any state that is not a synchronized\_pulse\_state of the system,
there are at least two synchronized clusters whose unified set of
nodes is not a synchronized set but that are within absorbance
distance of each other, and consequently they absorb. Thus,
eventually all synchronized clusters will comprise a synchronized
set of nodes.

We claim that if the following relationship between $REF$ and
$\Cycle$ is satisfied, then absorbance (of two synchronized clusters
whose unified set is not a synchronized set), is ensured
irrespective of the states of the synchronized clusters. Let
$C_{j^\prime}$ denote the largest synchronized cluster. The theorem
below, Theorem~\ref{thm:t2}, shows that for a given clustering of
$n-f$ correct nodes into $c>1$ synchronized clusters and for $n,$
$f,$ $\Cycle$ and $REF$ that satisfy

\ignore{-2mm}
\begin{equation} \sum_{j=1,j\ne j^\prime}^{c}ad(C_j)+\frac{1}{1-\rho}\sum_{g=1}^{n_{j^\prime}}\!\!R_g \ge \frac{1}{1-\rho}\cdot \Cycle\enspace
\label{eq:e14}
\end{equation}

\noindent there exist at least two synchronized clusters, whose
unified set is not a synchronized set of nodes, that will eventually
undergo absorbance.\\

Note that Eq.~\ref{eq:e14} is derived from Property 7
(Eq.~\ref{eq:prop7}):

Eq.~\ref{eq:prop7} derives the following equation (since the $R_g$
values are non-negative),

\begin{equation} \sum_{j=1,j\ne j^\prime}^{c}\sum_{g=f+1}^{f+n_j}\!\!\!R_g+\frac{1}{1-\rho}\sum_{g=1}^{n_{j^\prime}}\!\!R_g \ge
\frac{1}{1-\rho}\cdot \Cycle \enspace . \label{eq:e141}
\end{equation}

Incorporating the absorbance distance of Definition~\ref{def:ad}
into Eq.~\ref{eq:e141} yields exactly Eq.~\ref{eq:e14}. We use
Eq.~\ref{eq:e14} in Theorem~\ref{thm:t2} instead of
Eq.~\ref{eq:prop7} (Property 7) for readability of the proof.

\begin{theorem}\label{thm:t2} (Absorbance) Assume a clustering of $n-f$
correct nodes into $c>1$ synchronized clusters at real-time $t_0.$
Further assume that Eq.~\ref{eq:e14} holds for the resulting
clustering. Then there will be at least one synchronized cluster
that will absorb some other synchronized cluster by real-time
$t_0+2\cdot \cycle.$
\end{theorem}

\begin{proof}
Note that following the synchronized cluster procedure, the unified
set of the two synchronized clusters that will be shown to absorb,
are not necessarily a synchronized set of nodes at time $t_0.$
Assume without loss of generality that $C_{j^\prime}$ is the
synchronized cluster with the largest number of nodes, consequent to
running the clustering procedure. Exactly one out of the following
two possibilities takes place at $t_0$:

\begin{enumerate}
\item$\exists i\  (1 \le i\le c),$ such that $dist(C_i
,C_{(i+1)\!\!\!\pmod{c}}, t_0) \le ad(C_i ).$

\item $\forall i\ (1\le i\le c, i\neq j^\prime), dist(C_i ,C_{(i+1)\!\!\!\pmod{c}}, t_0) >
ad(C_i ).$
\end{enumerate}

Consider case 1. Following the protocol, $C_i$ must fire within
$\Cycle$ local-time units of $t_0.$ Observe the first real-time,
denoted $t_i,$ at which $C_i$ fires subsequent to real-time $t_0.$
Assume that $k \ge 0$ is the number of distinct inputs that causes
the Counter of at least one node in $C_i$ to reach the threshold and
fire (not counting the fire from nodes in $C_i$ itself). If $k>f$
then at least one correct node outside of $C_i$ caused some node in
$C_i$ to fire. This correct node must belong to some synchronized
cluster which is not $C_i.$ We denote this synchronized cluster
$C_x$ as its identity is irrelevant for the sake of the argument. We
assumed that at least one node in $C_i$ fired due to a node in
$C_x.$ Following Lemma~\ref{lem:l1} the rest of the nodes in $C_i.$
will follow as well, as a synchronized cluster is in particular a
synchronized set of nodes. This yields a new $dist(C_x ,C_i,\;..)$
of at most $3d.$ Following the same arguments as in Step 2 of
Lemma~\ref{lem:abs_lem}, $C_x$ and $C_i$ hence absorb. Therefore the
objective is reached. Hence assume that $k\le f$ and that $C_i$ did
not absorb with any preceding synchronized cluster. Thus, the last
real-time that $C_{(i+1)\!\!\!\pmod{c}}$ fired, denoted
$t_{C_{i+1}-fired},$ was before or equal to real-time $t_0,$ i.e.
$t_{C_{i+1}-fired} \le t_0 \le t_i$ and $dist(C_i
,C_{(i+1)\!\!\!\pmod{c}}, t_0) \le ad(C_i ).$ By
Theorem~\ref{thm:abs_thm}, $C_i$ will absorb
$C_{(i+1)\!\!\!\pmod{c}}.$\\

Consider case 2. We do not assume that $dist(C_{j^\prime}
,C_{(j^\prime+1)\!\!\!\pmod{c}}, t_0) > ad(C_{j^\prime}).$ Assume
that there is no absorbance until $C_{j^\prime}$ fires (otherwise
the claim is proven). Let $t_{j^\prime}$ denote the real-time at
which the first node in $C_{j^\prime}$ fires, at which
$\phi_{c_{j^\prime}}(t_{j^\prime})=0.$ There are two possibilities
at $t_{j^\prime}$:
\begin{itemize}
\item[2a.] $\exists i (1\le  i\le c),$ such that at $t_{j^\prime},
dist(C_i,C_{(i+1)\!\!\!\pmod{c}}, t_{j^\prime})\le ad(C_i ).$

\item[2b.] $\forall i (1\le i\le c, i\neq {j^\prime}), dist(C_i ,C_{(i+1)\!\!\!\pmod{c}}, t_{j^\prime}) >
ad(C_i ).$
\end{itemize}

Consider case 2a. This case is equivalent to case 1. The last
real-time that $C_{(i+1)\!\!\!\pmod{c}}$ fired, denoted
$t_{C_{i+1}-fired},$ was before or equal to real-time
$t_{j^\prime}.$ Denote $t_i$ the real-time at which the first node
of $C_i$ fires. Thus, $t_{C_{i+1}-fired} \le t_{j^\prime} \le t_i$
and $dist(C_i ,C_{(i+1)\!\!\!\pmod{c}}, t_0)\le ad(C_i ).$ By
Theorem~\ref{thm:abs_thm}, $C_i$ will absorb
$C_{(i+1)\!\!\!\pmod{c}}.$\\

Consider case 2b. We wish to calculate
$\phi_{c_{j^\prime+1}}(t_{j^\prime})$ and from this deduce the upper
bound on the threshold level of the first node in
$C_{(j^\prime+1)\!\!\!\pmod{c}}$ at real-time $t_{j^\prime}.$ We
first want to point out that
\begin{eqnarray}
\phi_{c_{j^\prime+1}}(t_{j^\prime})>\sum_{j=1,j\ne
j^\prime}^{c}ad(C_j) \; . \label{eq:calc_phi_ciplus1}
\end{eqnarray}
This stems from the fact that $C_{j^\prime}$ has just fired and that
$C_{j^\prime}$ and $C_{(j^\prime+1)\!\!\!\pmod{c}}$ are adjacent
synchronized clusters which implies that $$\forall i (1\le i\le c,
i\neq j^\prime\!+\!1), \phi_{c_{j^\prime+1}}(t_{j^\prime}) >
\phi_{c_i}(t_{j^\prime}).$$ Recall that
$\phi_{c_{j^\prime}}(t_{j^\prime})=0.$ From the case considered in
2b we have that $$\forall i (1\le i\le c, i\neq
{j^\prime}),\;dist(C_i ,\;C_{(i+1)\!\!\!\pmod{c}}, t_{j^\prime})
> ad(C_i ).$$
Thus Eq.~\ref{eq:calc_phi_ciplus1} follows. Following
Eq.~\ref{eq:e14} and Eq.~\ref{eq:calc_phi_ciplus1} we get:

\begin{eqnarray}
\phi_{c_{j^\prime+1}}(t_{j^\prime})>\sum_{j=1,j\ne
j^\prime}^{c}ad(C_j) \ge \frac{1}{1-\rho}\cdot \Cycle\enspace -
\frac{1}{1-\rho}\sum_{g=1}^{n_{j^\prime}}\!\!R_g\; .
\label{eq:calc_phi_cjplus1}
\end{eqnarray}

In the worst case the nodes of $C_{(j^\prime+1)\!\!\!\pmod{c}}$ are
slow. Thus at real-time $t_{j^\prime}$ they have measured, from
their last pulse, at least $(1-\rho)\cdot
\phi_{c_{j^\prime+1}}(t_{j^\prime}) = (1-\rho)\cdot
[\frac{1}{1-\rho}\cdot \Cycle\enspace -
\frac{1}{1-\rho}\sum_{g=1}^{n_{j^\prime}}\!\!R_g]=\sum_{g=n_{j^\prime+1}}^{n+1}\!\!R_g$
local-time units. Thus, following Property 1, the first node in
$C_{(j^\prime+1)\!\!\!\pmod{c}}$ receives the fire from
$C_{j^\prime}$ and increment its Counter to at least $n_{j^\prime}$
in a threshold level which is less or equal to $n_{j^\prime}$ and
will thus fire as well. Following Lemma~\ref{lem:l1} the rest of the
synchronized cluster will follow as well. This yields a new
$dist(C_{j^\prime} ,C_{(j^\prime+1)\!\!\!\pmod{c}},\;..)$ of at most
$3d.$ Following the same arguments as in Step 2 of
Lemma~\ref{lem:abs_lem}, $C_{j^\prime}$ and
$C_{(j^\prime+1)\!\!\!\pmod{c}}$ hence absorb.

Thus at least two synchronized clusters will absorb within $2\cdot
\cycle$ of $t_0$ which concludes the proof.  \end{proof}

The following theorem assumes the worst case of $n=3f+1.$

\begin{theorem}\label{thm:convergence} (Convergence) Within at most $2(2f+1)\!\cdot \cycle$ real-time units
the system reaches a synchronized\_pulse\_state.
\end{theorem}
\begin{proof}
Assume that $n=3f+1.$ Thus, the maximal number of synchronized
clusters is $2f+1,$ and since following Theorem~\ref{thm:t2} at
least two synchronized clusters absorb in every two cycles we obtain
the bound.
\end{proof}

\section{Analysis of the Algorithm and Comparison to Related Algorithms}
\label{sec:s5}\ignore{-2mm}

The protocol operates in two epochs: In the first epoch there is no
limitations on the number of failures and faulty nodes. In this
epoch the system might be in any state. In the second epoch there
are at most $f$ nodes that may behave arbitrarily at the same time,
from which the protocol may start to converge. Nodes may fail and
recover and nodes that have just recovered need time to synchronize.
Therefore, we assume that eventually we have a window of time within
which the turnover between faulty and non-faulty nodes is
sufficiently low and within which the system inevitably converges
(Theorem~\ref{thm:t2}).

{\bf Authentication and fault ratio:} The algorithm does not require
the power of unforgeable signatures, only an equivalence to an
authenticated channel is required. Note that the shared memory model
(\cite{DolWelSSBYZCS04}) has an implicit assumption that is
equivalent to an authenticated channel, since a node ``knows'' the
identity of the node that wrote to the memory it reads from. A
similar assumption is also implicit in many message passing models
by assuming a direct link among neighbors, and as a result, a node
``knows'' the identity of the sender of a message it receives.

Many fundamental problems in distributed networks have been proven
to require $3f+1$ nodes to overcome $f$ concurrent Byzantine faults
in order to reach a deterministic solution without
authentication~\cite{Impossibility86,R9,DHSS95,R12}. We have not
shown this relationship to be a necessary requirement for solving
the ``Pulse Synchronization'' problem but the results for related
problems lead us to believe that a similar result should exist for
the ``Pulse Synchronization'' problem.

There are algorithms that have no lower bound on the number of nodes
required to handle $f$ Byzantine faults, but unforgeable signatures
are required as all the signatures in the message are validated by
the receiver~\cite{DHSS95}. This is costly time-wise, it increases
the message size, and it introduces other limitations, which our
algorithm does not have. Moreover, within the self-stabilizing
paradigm, using digital signatures to counter Byzantine nodes
exposes the protocols to ``replay-attack'' which might empty its
usefulness.

{\bf  Convergence time:} We have shown in \cite{DDPBYZ-CS03} that
self-stabilizing Byzantine clock synchronization can be derived from
self-stabilizing Byzantine pulse synchronization. Conversely,
self-stabilizing Byzantine clock synchronization can be used to
trivially produce self-stabilizing Byzantine pulse synchronization.
Thus the two problems are supposedly equally hard. The only
self-stabilizing Byzantine clock synchronization algorithms besides
 \cite{DDPBYZ-CS03} are found
in~\cite{DolWelSSBYZCS04}. The randomized self-stabilizing Byzantine
clock synchronization algorithm
 published there synchronizes in $M\cdot2^{2(n-f)}$ steps,
where $M$ is the upper bound on the clock values held by individual
processors. The algorithm uses message passing, it allows transient
and permanent faults during convergence, requires at least $3f+1$
processors, but utilizes a global pulse system. An additional
algorithm in~\cite{DolWelSSBYZCS04}, does not use a global pulse
system and is thus partially synchronous similar to our model. The
convergence time of the latter algorithm is $O((n-f){n^{6(n-f)}}).$
This is drastically higher than our result, which has a cycle length
of $O(f^2)\cdot d$ time units and converges within $2(2f+1)$ cycles.
The convergence time of the only other correct self-stabilizing
Byzantine pulse synchronization algorithm \cite{new-pulse-tr} has a
cycle length of $O(f)\cdot d$ time units and converges within 6
cycles.

{\bf Message and space complexity: } The size of each message is
$O(logn)$ bits. Each correct node multicasts exactly one message per
cycle. This yields a message complexity of at most $n$ messages per
cycle. The system's message complexity to reach synchronization from
any arbitrary state is at most $2n(2f+1)$ messages per
synchronization from any arbitrary initial state. The faulty nodes
cannot cause the correct nodes to fire more messages during a cycle.
Comparatively, the self-stabilizing clock synchronization algorithm
in~\cite{DolWelSSBYZCS04} sends $n$ messages during a pulse and thus
has a message complexity of $O(n(n-f)n^{6(n-f)}).$ This is
significantly larger than our message complexity irrespective of the
time interval between the pulses. The message complexity of the only
other correct self-stabilizing Byzantine pulse synchronization
\cite{new-pulse-tr} equals $O(n^3)$ per cycle.

The space complexity is $O(n)$ since the variables maintained by the
processors keep only a linear number of messages recently received
and various other small range variables. The number of possible
states of a node is linear in $n$ and the node does not need to keep
a configuration table.

The message broadcast assumptions, in which every message, even from
a faulty node, eventually arrives at all correct nodes, still leaves
the faulty nodes with certain powers of multifaced behavior since we
assume nothing on the order of arrival of the messages. Consecutive
messages received from the same source within a short time window
are ignored, thus, a faulty node can send two concomitant messages
with differing values such that two correct nodes might receive and
relate to different values from the same faulty node.

{\bf Tightness of synchronization:} In the presented algorithm, the
invocation of the pulses of the nodes will be synchronized to within
the bound on the relay time of messages sent and received by correct
nodes. In the broadcast version, this bound on the relay time equals
$d$ real-time units. Note that the lower bound on clock
synchronization in completely connected, fault-free
networks~\cite{R21} is $d(1-1/n).$ We have shown in
Section~\ref{sec:IAGree} how the algorithm can be executed in
non-broadcast networks to achieve a synchronization tightness of
$\sigma=3d.$ Comparatively, the clock synchronization algorithm
of~\cite{DHSS95} reaches a synchronization tightness typical of
clock synchronization algorithms of $d(1+\rho) + 2\rho(1+\rho)\cdot
R,$ where $R$ is the time between re-synchronizations. The second
Byzantine clock synchronization algorithm in~\cite{DolWelSSBYZCS04}
reaches a synchronization tightness which is in the magnitude of
$(n-f)\cdot d(1+\rho).$ This is significantly less tight than our
result. The tightness of the self-stabilizing Byzantine pulse
synchronization in \cite{new-pulse-tr} equals $3d$ real-time units.

{\bf Firing frequency bound:} The firing frequency upper bound
during normal steady-state behavior is around twice that of the
endogenous firing frequency of the nodes. This is because $\cyclemin
\ge \frac{\Cycle}{2}.$ This bound is influenced by the fraction of
faulty nodes (the sum of the first $f$ threshold steps relative to
\Cycle). For $n=3f+1$ this translates to $\approx \frac{1}{2}
\Cycle.$ Thus, if required, the firing frequency bound can be closer
to the endogenous firing frequency of $1\cdot \Cycle$ if the
fraction of faulty nodes is assumed to be lower. For example, for a
fraction of fault nodes of $f=\frac{n}{10},$ the lower bound on the
cycle length, $\cyclemin, $ becomes approximately $8/9$ that of the
endogenous cycle length. $\cyclemax = \Cycle\cdot (1+\rho)$
real-time units.

\section{Discussion}\label{sec:discussion}\ignore{-2mm}
We developed and presented the ``Pulse Synchronization'' problem in
general, and an efficient linear-time self-stabilizing Byzantine
pulse synchronization algorithm, {\sc bio-pulse-synch}, as a
solution in particular. The pulse synchronization problem poses the
nodes with the challenge of invoking regular events synchronously.
The system may be in an arbitrary state in which there can be an
unbounded number of Byzantine faults. The problem requires the
pulses to eventually synchronize from any initial state
 once the bound on the permanent number of Byzantine
failures is less than a third of the network. The problem resembles
the clock synchronization problem though there is no ``value'' (e.g.
clock time) to agree on, rather an event in time. Furthermore, to
the best of our knowledge, the only efficient self-stabilizing
Byzantine clock synchronization algorithm assumes a background pulse
synchronization module.

The algorithm developed is inspired by and shares properties with
the lobster cardiac pacemaker network; the network elements (the
neurons) fire in tight synchrony within each other, whereas the
synchronized firing pace can vary, up to a certain extent, within a
linear envelope of a completely regular firing pattern.

\ignore{A self-stabilizing pulse synchronization algorithm can be
beneficial to solve TDMA based slot synchronization (in multi access
channels such as cellular phones; slotted ALOHA)~\cite{R10}. A
similar suggestion appears in~\cite{R5}, though without supplying
any fault tolerance. In this problem, time is divided into frames of
fixed length, and each frame contains a certain number of slots.
Frames are repeated periodically, and each participant is
exclusively allocated one slot per frame. It is of paramount
importance to avoid collisions of the slot allocations. If the
participants are not tightly synchronized with respect to the frame
starting point then collisions will occur. The exact starting time
of the frame is of minor importance relative to the synchronization
requirement. Should a central synchronizer be present then the
network would be extremely sensitive to faults in the central
station. Where no central station is available and connectivity
changes rapidly a self-stabilizing distributed synchronization
protocol would make the participants resynchronize without any
exterior intervention, should the participant be out of
synchronization.  Other cases that could benefit from a
self-stabilizing pulse synchronization scheme include global
snapshots; load balancing; distributed inventory count; distributed
scheduling; backup of distributed systems; debugging of distributed
programs; and deadlock detection.}

A number of papers have recently postulated on the similarity
between elements connected with biological robustness and design
principles in engineering \cite{AlonNature99,KitanoNature04}. In the
current paper we have observed and understood the mechanisms for
robustness in a comprehensible and vital biological system and shown
how to make specific use of analogies of these elements in
distributed systems in order to attain high robustness in a
practical manner. The same level of robustness has not been
practically achieved earlier in distributed systems. We postulate
that our result elucidates the feasibility and adds a solid brick to
the motivation to search for and to understand biological mechanisms
for robustness that can be carried over to computer systems.

The neural network simulator SONN (\cite{SONN}) was used in early
stages of developing the algorithm for verification of the protocol
in the face of probabilistic faults and random initial states. It is
worth noting that the previous pulse synchronization procedure found
in \cite{DDPBYZ-CS03} was mechanically verified at NASA LaRC
(\cite{NASA-BYZSS}) which greatly facilitated uncovering its flaw. A
natural next step should thus be to undergo simulation and
mechanical verification of the current protocol that can mimic a
true distributed system facing transient and Byzantine faults.

\ignore{{\bf Possible biological implications:} We may now speculate
on possible new meanings of the refractory function in the
biological context. Could there be some correlation between the
number of faults and their severity that the biological neuronal
network must face and the network size required to reach
synchronization? Could there be some relationship between the shape
of the neuronal refractory function and the size of the biological
neural network? Does the shape of the refractory function of a
single neuron play a role in the networks ability to tolerate faults
or noise? Does the shape of the refractory function of a single
neuron play a role in the networks ability to alter the synchronized
firing pace?}

\ignore{
\noindent \textbf{Acknowledgments.} The authors wish to thank the
anonymous reviewers for helpful comments that significantly improved
the presentation of this paper.
}

%%%%%%%%%%%%%%%%%%%%%%%%%%%%%%%%%%%%%%%%%%%%%%%%
%\newpage
%%%%%%%%%%%%%%%%%%%%%%%%%%%%%%%%%%%%%%%%%%%%%%%%

%\newpage

%\appendix
\section{Appendix}\label{sec:apdx}

% end ignore

%\newpage
\noindent{\bf Proof of correctness of the {\sc summation} procedure:}\\

\begin{lemma}\label{lem:powerseries} For $k \in \mathbb{N},\;k \ge 0,$ $$\tau(k)\cdot \frac{1+\rho}{1-\rho}+2d(1+\rho)=\tau(k+1)\enspace .$$

\end{lemma}

\begin{proof} \[ \tau(k)\cdot \frac{1+\rho}{1-\rho}+2d(1+\rho) =
[2d(1+\rho)\frac{(\frac{1+\rho}{1-\rho})^{k+1}-1}{(\frac{1+\rho}{1-\rho})-1}]\cdot
\frac{1+\rho}{1-\rho}+2d(1+\rho) \]
\[ = [2d(1+\rho)\sum_{i=0}^{k}(\frac{1+\rho}{1-\rho})^i]\cdot
\frac{1+\rho}{1-\rho}+2d(1+\rho) =
[2d(1+\rho)\sum_{i=1}^{k+1}(\frac{1+\rho}{1-\rho})^i]+2d(1+\rho)
\]
\[ = 2d(1+\rho)\sum_{i=0}^{k+1}(\frac{1+\rho}{1-\rho})^i  = 2d(1+\rho)\frac{(\frac{1+\rho}{1-\rho})^{k+2}-1}{(\frac{1+\rho}{1-\rho})-1} = \tau(k+1)
\enspace . \tre  \]

\end{proof}

\begin{lemma}\label{lem:messageAgeTau} Let a correct node $q$ receive a message
$M_p$ from a correct node $p$ at local-time $t_{arr}.$ For every one
of $p$'s stored messages $(S_r, t^\prime)$ that is accounted for in
$Counter_{M_p},$ then at $q,$ from some time $t$ in the local-time
interval $[t_{arr}, t_{arr} + d(1+\rho )]$ and at least until the
end of the interval:

$$MessageAge(t, q, r) \le \tau(Counter_{M_p}+1)\;.$$

\end{lemma}

\begin{proof} Following the {\sc prune} procedure at $p,$ the oldest of its stored
messages accounted for in $Counter_{M_p}$ was at most
$\tau(Counter_{M_p})$ time units old on $p$'s clock at the time it
sent $M_p.$ This oldest stored message could have arrived at $q,$
$\delta(1+\rho )$ local-time units on $q$'s clock, prior to its
arrival at $p.$ Within this time $p$ should also have received all
the messages accounted for in $M_p.$ Another $\pi(1+\rho )$
local-time units could then have passed on $q$'s clock until $M_p$
was sent. $M_p$ could have arrived at $q,$ $\delta(1+\rho )$ time
units on $q$'s clock after it was sent by $p.$ By this time $q$
would also have received all the messages that are accounted for in
$M_p,$ irrespective if $q$ had previous messages from the same
nodes. Another $\pi(1+\rho )$ time units can then pass on $q$'s
clock until all messages are processed. Thus, in the worst case that
node $p$ is slow and node $q$ is fast and by
Lemma~\ref{lem:powerseries}, for every stored message accounted for
in $Counter_{M_p}, \exists t\in [t_{arr}+d(1+\rho)],$ we have:

\vspace{-6mm}\begin{eqnarray*}
\lefteqn{MessageAge(t, q, r) \le MessageAge(t_{arr}+d(1+\rho), q, r)}\\
&&\hspace{1cm}
 \le \tau(Counter_{M_p})\cdot\frac{1+\rho}{1-\rho}+\delta(1+\rho ) +
\pi(1+\rho ) +
\delta(1+\rho ) + \pi(1+\rho ) \\
&&\hspace{1cm}
=\tau(Counter_{M_p})\cdot\frac{1+\rho}{1-\rho}+2d(1+\rho)=\tau(Counter_{M_p}+1)\enspace
.
\end{eqnarray*}

\mbox{\ } \end{proof}

\begin{lemma}\label{lem:l4} The Counter of a correct node cannot
exceed $n$ and a correct node will not send a Counter that exceeds
$n-1.$
\end{lemma}
\begin{proof}
There can be at most $n$ distinct stored messages in the CS of a
correct node hereby bounding the Counter by $n.$

For a correct node to have a Counter that equals exactly $n$ it
needs its own stored message to be in its CS, as a consequence of a
message it sent. Consider the moment after it sent this message, say
before the node's Counter reached $n,$ that is accounted for in its
CS. This message was concomitant to its pulse invocation and cycle
reset. The node assesses its own message at most $d(1+\rho)$
local-time units after sending it thus, following the {\sc prune}
procedure, its own stored message will decay at most $\tau(n+2) +
d(1+\rho ) < \tau(n+3))=R_{n+1}$ local-time units after it was sent.
Thus at the moment the node reaches threshold level $R_n$ its own
message will already have decayed and the Counter will decrease and
will be at most $n-1,$ implying that any message sent by the node
can carry a Counter of at most $n-1.$  \end{proof}

\begin{lemma}\label{lem:NotInRUCS} A stored message, $(S_r, t^\prime),$
that has been moved to the RUCS of a correct node $q$ up to
$d(1+\rho)$ local-time units subsequent to the event of sending a
message $M_p$ by $p,$ (or was moved at an earlier time) cannot have
been accounted for in $Counter_{M_p}.$
\end{lemma}

\begin{proof}
Assume that the stored message $(S_r, t^\prime)$ was moved to the
{\sc RUCS} of node $q$ at a local-time $t,$ $d(1+\rho)$ local-time
units subsequent to the event $t_{{\mathrm send}\,M_p}$ at node $p,$
(or it was moved at an earlier time). Thus at $q$ at local-time $t,$
$MessageAge(t, q, r)
> \tau(n+1).$ Therefore at node $p$ at local-time $t_{{\mathrm send}\,M_p},$
$MessageAge(t_{{\mathrm send}\,M_p}, p, r) > \tau(n+1) - 2d(1+\rho)>
\tau(n). $ This is because $p$ could have received the message $M_r$
up to $d(1+\rho)$ local-time units later than $q$ did, and $q$ could
have received $M_p$ up to $d(1+\rho)$ local-time units after it was
sent.

Following the {\sc prune} procedure at $p,$ $(S_r, t")$ would have
been accounted for at the sending time of $M_p$ only if
$Counter_{M_p}\ge  n+1.$ Therefore by Lemma~\ref{lem:l4} node $p$
did not account for the stored message of $r$ in $Counter_{M_p}.$
 \end{proof}

\begin{corollary}\label{lem:l6} A stored message, $(S_r, t^\prime),$ that has decayed at a
correct node $q$ prior to the event of sending a message $M_p$ by
$p,$ cannot have been accounted for in $Counter_{M_p}.$
\end{corollary}

\begin{proof}
Corollary~\ref{lem:l6} is an immediate corollary of
Lemma~\ref{lem:NotInRUCS}.  \end{proof}

\begin{corollary}\label{cor:stateSize} Let a correct node $q$ receive a message
$M_p$ from a correct node $p$ at local-time $t_{arr}.$ Then, at $q,$
from some time $t$ in the local-time interval $[t_{arr}, t_{arr} +
d(1+\rho )]$ and at least until the end of the interval:

$$\|\MessagePool\| \ge Counter_{M_p}+1\enspace .$$

\end{corollary}

\begin{proof}
Corollary~\ref{cor:stateSize} is an immediate corollary of
Lemma~\ref{lem:messageAgeTau} and Lemma~\ref{lem:NotInRUCS}.
 \end{proof}

Thus, as a consequence to the lemmata, we can say informally, that
when the system is coherent all correct nodes relate to the same set
of messages sent and received.\\

\subsection{Proof of Theorem~\ref{thm:timely}}
\label{ssec:proof_timely} Recall the statement of
Theorem~\ref{thm:timely}:

\emph{Any message, $M_p,$ sent by a correct node $p$ will be
assessed as
timely by every correct node $q.$}\\

\begin{proof} Let $M_p$ be sent by a correct node $p,$ and received by a correct
node $q$ at local-time $t_{arr.}$ We show that the timeliness
conditions hold:

\noindent Timeliness Condition 1: $0\le  Counter_{M_p} \le n-1$ as
implied by Lemma~\ref{lem:l4} and by the fact that the CS cannot
hold a negative number of stored messages.\\

\noindent Timeliness Condition 2: Following Lemma~\ref{lem:l4} a
correct node will not fire during the absolute refractory period.
Property 5 therefore implies that a correct node cannot count less
than $\tau(n+3)$ local-time units between its consecutive firings. A
previous message from a correct node will therefore be at least
$\tau(n+2)$ local-time units old at any other correct node before it
will receive an additional message from that same node. Following
the {\sc prune} procedure, the former message will therefore have
decayed at all correct nodes and therefore cannot be present in the
$\MessagePool$ at the arrival time of the subsequent
message from the same sender.\\

\noindent Timeliness Condition 3: This timeliness condition
validates $Counter_{M_p}.$ The validation criterion relies on the
relation imposed at the sending node by the {\sc prune} procedure,
between the $MessageAge(t, p, ..)$ of its accounted stored messages
and its current Counter.

By Lemma~\ref{lem:messageAgeTau}, for all stored messages $(S_r,
t^\prime)$ accounted for in $M_p,$ \\$MessageAge(t, q, r) \le
\tau(Counter_{M_p}+1)$ from some local-time $t \in [t_{arr}, t_{arr}
+ d(1+\rho )]$ and until the end of the interval.

By Corollary~\ref{cor:stateSize}, $\|\MessagePool\| \ge
Counter_{M_p}+1,$ from some local-time $t^{\prime\prime} \in
[t_{arr}, t_{arr} + d(1+\rho )]$ and until the end of the interval.

We therefore proved that Timeliness Condition 3 holds for any $0\le
k<n$ at the latest at local-time $t_{arr} + d(1+\rho).$\\

\noindent The message $M_p$ is therefore assessed as timely by $q.$
\end{proof}

\begin{lemma}\label{lem:l8} Following the arrival and assessment of a
timely message $M_p$ at node $q,$ the subsequent execution of the
{\sc make-accountable} procedure yields $Counter_q > Counter_{M_p}.$
\end{lemma}

\begin{proof}
We first show that at time $t,$ the time of execution of the {\sc
make-accountable} procedure, $\max[1,\; (Counter_{M_p} - Counter_q +
1)]\le  \|\mbox{UCS}\|,$ ensuring the existence of a sufficient
number of stored messages in UCS to be moved to CS.

$M_p$ is assessed as timely at $q,$ therefore, by Timeliness
Condition 3 and Lemma~\ref{lem:NotInRUCS}, at time $t,$

\begin{eqnarray*}
Counter_{M_p} &<& \|\MessagePool\| = \|\mbox{CS}\| +
\|\mbox{UCS}\| = Counter_q + \|\mbox{UCS}\| =\\
&&Counter_{M_p} - \max[1,\; (Counter_{M_p} - Counter_q + 1)] + 1 +
\|\mbox{UCS}\|\\
&\Rightarrow&  0 < - \max[1,\; (Counter_{M_p} - Counter_q + 1)] + 1
+
\|\mbox{UCS}\|\\
&\Rightarrow&   \max[1,\; (Counter_{M_p} - Counter_q + 1)] - 1 <
\|\mbox{UCS}\|\\
&\Rightarrow&   \max[1,\; (Counter_{M_p} - Counter_q + 1)] \le
\|\mbox{UCS}\|\enspace .
\end{eqnarray*}

There are two possibilities at the instant {\bf prior} to the
execution of the {\sc make-accountable} procedure. At this instant
$Counter_q = \|\mbox{CS}\|$:
\begin{enumerate}
\item $Counter_{M_p} \le Counter_q,$ then $\max[1,\; (Counter_{M_p} -
Counter_q + 1)] = 1,$ meaning $\|\mbox{CS}\|$ will increase by $1.$

\item
 $Counter_{M_p} > Counter_q,$ then $\|\mbox{CS}\|$ will be
 $Counter_q + \max[1,\; (Counter_{M_p} - Counter_q + 1)] = Counter_q + Counter_{M_p}
 - Counter_q + 1 = Counter_{M_p} + 1.$
\end{enumerate}
In either case, immediately subsequent to the execution of the
procedure we get: $\|\mbox{CS}\| > Counter_{M_p}$ and therefore the
updated $Counter_q > Counter_{M_p}. $  \end{proof}

\subsection{Proof of Lemma~\ref{lem:counter_inc}}
\label{ssec:proof_counter_inc} Recall the statement of
Lemma~\ref{lem:counter_inc}:

\emph{Following the arrival of a timely message $M_p,$ at a node
$q,$ then at time $t_{\scriptsize{\mathrm send}\,M_q},$ $Counter_q
> Counter_{M_p}.$}\\

\begin{proof} Let $t_{arr}$ denote the local-time of arrival of $M_p$ at $q.$
Recall that $t_{\scriptsize{\mathrm send}\,M_q}$ is the local-time
at which $q$ is ready to assess whether to send a message consequent
to the arrival and processing of $M_p.$ In the local-time interval
$[t_{arr}, t_{\scriptsize{\mathrm send}\,M_q}]$ at least one {\sc
prune} procedure is executed at $q$, the one which is triggered by
the arrival of $M_p.$ Following Lemma~\ref{lem:l8}, $Counter_q
> Counter_{M_p}$ subsequent to the execution of the {\sc
make-accountable} procedure. Note that $t_{arr} \le t_{send~M_q} \le
t_{arr}+d(1+\rho).$ By Lemma~\ref{lem:NotInRUCS} all stored messages
accounted for in $Counter_{M_p}$ will not be moved out of the
$\MessagePool$ by any {\sc prune} procedure executed up to
local-time $t_{{\mathrm send}\,M_q},$ thus, $Counter_q$ must stay
with a value greater than $Counter_{M_p}$ up to time $t_{{\mathrm
send}\,M_q}. $ \end{proof}

\subsection{Lemma~\ref{lem:l3}}
\begin{lemma}\label{lem:l3} Let $p, q\in C_i$ and $r\in C_j,$ denote three
correct nodes belonging to two different synchronized clusters.
Following the arrival and assessment of $p$'s and $q$'s fires, both
will be accounted for in the Counter of $r.$
\end{lemma}

\begin{proof}
Without loss of generality, assume that $p$ fires before node $q.$
Following Lemma~\ref{lem:l1} node $q$ will fire within $\sigma$ of
$p$ ($d(1+\rho )$ on $r$'s clock). Node $r$ will receive and assess
$q$'s fire at a time $t_q$ at most $d(1+\rho ) + d(1+\rho ) =
2d(1+\rho )$ after $p$ fired. Summation Property [P2] ensures that
$r$ will account for each one after their arrival and assessments.
Furthermore, $MessageAge(t_q, q, p)\le 2d(1+\rho ) = \tau(0)$ and
therefore node $r$ did not decay or move $M_p$ to RUCS by time
$t_q.$ Therefore, $M_p$ is still accounted for by node $r$ at time
$t_q$ and thus, both $p$ and $q$ are accounted for in $Counter_r$ at
time $t_q.$
 \end{proof}

\end{document}